\documentclass[journal]{IEEEtran}
\usepackage[numbers]{natbib}
\usepackage[T1]{fontenc}
\usepackage[utf8]{inputenc} 
\usepackage{hyperref}              % customize internal links 
\usepackage{amsmath}
\usepackage{amssymb}
\usepackage[output-decimal-marker={,},exponent-product=\cdot]{siunitx} % for big numbers
\usepackage{amsthm}
\usepackage{amsfonts}
\usepackage[mathscr]{eucal}        % lettres calligraphiques différentes de \mathcal
\usepackage{csquotes}
\usepackage{fancyhdr}              % headers and footers
\usepackage{enumitem}              % customize items
\usepackage{multirow}
\usepackage{float}                 % forcer la position d'une image
\usepackage{pgfplots}
\usepackage{tikz}
\usetikzlibrary{arrows,chains,matrix,positioning,scopes}
\usepackage[all]{xy}               % beaux diagrammes
\usepackage[ruled,vlined,linesnumbered]{algorithm2e}
\usepackage{color, colortbl}
\usepackage[margin=0.5cm]{caption}
\usepackage[lofdepth,lotdepth]{subfig}
\usepackage{hyphenat}              % pour les césures de mots avec un tiret (utiliser \hyp{} pour le tiret)
\usepackage{mathtools}
\usepackage{dsfont}                % pour le 'double 1' (fonction indicatrice) 
\usepackage{listings}
\usepackage{arydshln}
\usepackage{multirow}
\theoremstyle{plain}
\newtheorem{theorem}{Theorem}[section]
\newtheorem{proposition}[theorem]{Proposition}
\newtheorem{corollary}[theorem]{Corollary}
\newtheorem{lemma}[theorem]{Lemma}
\newtheorem*{lemma*}{Lemma}

\newtheorem{problem}[theorem]{Problem}

\theoremstyle{definition}
\newtheorem{definition}[theorem]{Definition}
\newtheorem{remark}[theorem]{Remark}
\newtheorem{example}[theorem]{Example}
\newtheorem{construction}[theorem]{Construction}

\title{Private Information Retrieval from Transversal Designs}
\date{\today}
\author{
  Julien \textsc{Lavauzelle}\\
  Laboratoire LIX, École Polytechnique, Inria \& CNRS UMR 7161\\
  Université Paris-Saclay
  \thanks{This paper appears in: IEEE Transactions on Information Theory, on pages: 1-17, DOI: \href{https://doi.org/10.1109/TIT.2018.2861747}{10.1109/TIT.2018.2861747}.}
  \thanks{It was presented in part at the Tenth International Workshop on Coding and Cryptography 2017, September 18-22, 2017, Saint-Petersburg, Russia.}
  \thanks{This work is partially funded by  French ANR-15-CE39-0013-01 \enquote{Manta}.}  
}
\newcommand\FF{\mathbb{F}}
\newcommand\ZZ{\mathbb{Z}}

\newcommand\PP{\mathbb{P}}
\renewcommand\AA{\mathbb{A}}

\newcommand\rank{\mathrm{rank}}

\newcommand\calA{\mathcal{A}}
\newcommand\calB{\mathcal{B}}
\newcommand\calG{\mathcal{G}}

\newcommand\calD{\mathcal{D}}
\newcommand\calC{\mathcal{C}}

\newcommand\calQ{\mathcal{Q}}

\newcommand\calR{\mathcal{R}}

\newcommand\calO{\mathcal{O}}
\newcommand\calT{\mathcal{T}}

\newcommand\Code{\mathrm{Code}}
\newcommand\IC{\mathrm{IC}}
\newcommand\TD{\mathrm{TD}}
\newcommand\tTD{t\text{-}\mathrm{TD}}
\newcommand\genTD[1]{#1\text{-}\mathrm{TD}}
\newcommand\AG{\mathrm{AG}}
\newcommand\PG{\mathrm{PG}}
\newcommand\OA{\mathrm{OA}}

\newcommand\parity{\mathcal{C}_{\textrm{par}}}
\newcommand\mydef{\coloneqq}
\begin{document}

\maketitle

% \IEEEoverridecommandlockouts
% \IEEEpubid{
%   \begin{minipage}{\textwidth}\ \\[12pt] \centering
%   0000--0000/00\$00.00~\copyright~2018 IEEE. Personal use of this material is permitted. However, permission to use this material for any other purposes must be obtained from the IEEE by sending a request to \url{pubs-permissions@ieee.org.}
%   \end{minipage}
% }

\begin{abstract}
  Private information retrieval (PIR) protocols allow a user to retrieve entries of a database without revealing the index of the desired item. Information-theoretical privacy can be achieved by the use of several servers and specific retrieval algorithms. Most known PIR protocols focus on decreasing the number of bits exchanged between the client and the server(s) during the retrieval process. On another side, Fazeli \emph{et al.} introduced so-called \emph{PIR codes} in order to reduce the storage overhead on the servers. However, few works address the issue of the computation complexity of the servers.

  It this paper, we show that a specific encoding of the database yields PIR protocols with reasonable communication complexity, low storage overhead and optimal computational complexity for the servers. This encoding is based on incidence matrices of transversal designs, from which a natural and efficient recovering algorithm is derived. We also present several instances for our construction, which make use of finite geometries and orthogonal arrays. We finally give a generalisation of our main construction in order to resist collusions of servers.
\end{abstract}

% \begin{keywords}
%   Private information retrieval, linear code, transversal design, orthogonal array, finite geometry
% \end{keywords}

\section{Introduction}

\subsection{Private Information Retrieval}

A private information retrieval (PIR) protocol aims at ensuring a user that he can retrieve some part $D_i$ of a remote database $D$ without revealing the index $i$ to the server(s) holding the database. For example, such protocols can be applied in medical data storage where physicians would be able to access parts of the genome while hiding the specific gene they analyse. The PIR paradigm was originally introduced by Chor, Goldreich, Kushilevitz and Sudan~\cite{ChorGKS95, ChorKGS98}.

% \IEEEpubidadjcol

A naive solution to the problem consists in downloading the entire database each time the user wants a single entry. But the communication complexity would then be overwhelming, so we look for PIR protocols exchanging less bits. However, Chor \emph{et al.} proved that, when the $k$-bits database is stored on a single server, a PIR protocol which leaks no information on the index $i$ (such a protocol being called \emph{information-theoretically secure}) must use $\Omega(k)$ bits of communication~\cite{ChorKGS98}. Two alternatives were then considered: restricting the protocol to computational security (initiated by Chor and Gilboa~\cite{ChorG97}), or allowing several servers to store the database. Our work focuses on the last one.

In many such PIR protocols the database is \emph{replicated} on $\ell$ servers, $\ell > 1$. Informally, the idea is that each server is asked to compute some partial information related to a random-like query sent by the user. Then the user collects all the servers' answers and retrieves the desired symbol with an appropriate algorithm. For instance, Chor \emph{et al.}~\cite{ChorKGS98} considered a smart arrangement of the database entries in a $\log(\ell)$-dimensional array, and used XOR properties to mask the index of the desired item and to retrieve the associated symbol. Their protocol features decreasing communication as a function of the number of servers: with $\ell$ servers, the communication is $\calO(\ell \log(\ell) k^{1/\log \ell})$ bits. For constant $\ell$, the authors also proposed a PIR protocol with communication $\calO(k^{1/\ell})$. A few years later, Katz and Trevisan~\cite{KatzT00} showed that any smooth locally decodable code $\calC \subseteq \Sigma^n$ of locality $\ell$ gives rise to a PIR protocol with $\ell$ servers whose communication complexity is $\calO(\ell \log(n |\Sigma|))$ --- see \cite{Yekhanin12} for a good survey on locally decodable codes (LDCs) and their applications in PIR protocols. Building on this idea, many PIR schemes (notably \cite{BeimelIKR02, Yekhanin08, Efremenko12, DvirG16}) successively decreased the communication complexity, achieving $\calO(k^{\sqrt{\log \log k /\log k}})$ with only $\ell = 2$ servers. However, only few of them tried to lighten the computational and storage cost on the server side.

% \IEEEpubidadjcol

By preprocessing the database, Beimel, Ishai and Malkin~\cite{BeimelIM04} were the first to address the minimization of the server storage/computation in PIR protocols. Then, initiated by Fazeli, Vardy and Yaakobi~\cite{FazeliVY15}, recent works used the concept of \emph{PIR codes} to address the storage issue. The idea is to turn an $\ell$-server replication-based PIR protocol into a more-than-$\ell$-server distributed PIR protocol with a smaller overall storage overhead. For this purpose, the user encodes the database and distributes pieces of the associated codeword among the servers, such that servers hold distinct parts of the database (plus some redundancy). Through this transformation, both communication complexity and computational cost keep the same order of magnitude, but the storage overhead corresponds to the PIR code's one, which can be brought arbitrarily close to $1$ when sufficiently many servers are used. Several recent works also address the PIR issue on previously coded databases~\cite{TajeddineR16}, and/or aim at reaching the so-called capacity of the model~\cite{SunJ17}. However, while the storage drawback seems to be solved, huge computational costs still represent a barrier to the practicality of such PIR protocols.

\subsection{Motivations and results}

As pointed out by Yekhanin~\cite{Yekhanin12}, \enquote{the overwhelming computational complexity of PIR schemes (...) currently presents the main bottleneck to their practical deployment}. Consider a public database which is frequently queried, \emph{e.g.} a database storing stock exchange prices where private queries could be very relevant. Fast retrieval is crucial is this context. Hence, one cannot afford each run of the PIR protocol to be computationally inefficient, for instance $\Omega(k)$ if $k$ is the size of the database. Therefore, a relevant goal is to build PIR protocols with sublinear computational complexity in the length of the database stored by each server.

Naively, the computational complexity of a PIR protocol could be drastically reduced if we let all possible answers to its queries to be precomputed. Of course, storing all these answers dramatically increases the needed storage, so let us focus on a construction due to Augot, Levy-dit-Vehel and Shikfa~\cite{AugotLS14} --- anterior to the PIR codes breakthrough~\cite{FazeliVY15} --- that address this issue.

The construction of Augot \emph{et al.}~\cite{AugotLS14} uses a specific family of high-rate locally decodable codes called \emph{multiplicity codes} introduced by Kopparty, Saraf and Yekhanin~\cite{KoppartySY14}. But instead of \emph{replicating} the database on $\ell$ servers ($\ell > 1$ being the locality of the codes), the authors \emph{split} an encoded version $c$ of the database $D$ into parts $c^{(1)},\dots, c^{(\ell)}$, and share these parts on the servers. The main difference with PIR codes~\cite{FazeliVY15} is that Augot \emph{et al.}'s construction does not purpose to \emph{emulate} a lighter PIR protocol with an existing one. It uses specific properties of the encoding as a way to split the database on several servers. In short, the multiplicity codes they use feature \emph{both} the privacy of the PIR protocol and the storage reduction for the servers. We refer to Section~\ref{sec:comparison} for more details on the construction.

In this work, we reconsider this \enquote{codeword support splitting} idea, and we propose a new generic framework for the construction of PIR protocols which takes into account the computational complexity issue. More precisely, the protocols we give are computationally optimal with respect to the communication complexity of the protocol, in the sense that each server needs to read \emph{only one} entry in the part of the database it holds.

Our construction is based on combinatorial structures called \emph{transversal designs}, from which we naturally derive a linear code, a partition of its support and a \emph{local} reconstruction algorithm. In practice, we give several instances of transversal designs that lead to codes with large rate, hence to PIR protocols with low storage overhead. The two first families come from incidences between points and lines in the affine (resp. projective) space. They are closely related to the classical geometric designs of $1$-flats. A third family of instances makes use of a classical transformation of so-called \emph{orthogonal arrays} of strength $2$ into transversal designs. We then proceed to a thorough study of the dimension of codes coming from \emph{MDS-like} orthogonal arrays of strength $2$. A fourth and last family of practical instances appears when showing that orthogonal arrays built from \emph{divisible codes} lead to PIR protocols with storage expansion less than $2$. We finally  prove that orthogonal arrays with strength $t > 2$ allow the construction of PIR protocols resisting to collusions of up to $t-1$ servers. We exhibit and analyzed instances of some orthogonal arrays with large strength to conclude this work.

\subsection{Organization}

We start by giving two formal definitions of PIR protocols in Section~\ref{sec:definitions}, depending on whether the database is replicated or distributed on the servers. We also present the standard construction of replication-based PIR protocols from smooth locally decodable codes. In Section~\ref{sec:designs}, we recall definitions of combinatorial structures and their associated codes. The $1$-private PIR protocols based on transversal designs are introduced in Section~\ref{sec:no-collusion-PIR}. Section~\ref{sec:explicit-constructions}  is devoted to four families of instances of the PIR construction having practical parameters. Finally, a generalisation of our construction is given in Section~\ref{sec:collusion-PIR} in order to keep up with collusions of servers, and a comparison with the PIR protocols coming from multiplicity codes is presented in Section~\ref{sec:comparison}.

\section{Definitions and related constructions}
\label{sec:definitions}

We first recall that we are only concerned with information-theoretically secure PIR protocols. In this paper, we denote by $U$ the \emph{user} (or \emph{client}) of the PIR protocol. User $U$ owns a database denoted by $D = (D_i)_{1 \le i \le k} \in \FF_q^k$, where $\FF_q$ represents the finite field with $q$ elements. Database $D$ hence contains $|D| = k \log q$ bits. We also denote by $S_1, \dots, S_{\ell}$ the $\ell$ servers involved in the PIR protocol.

Given $A$, $B$ two sets, with $|B| = n < \infty$, we denote by $A^B$ the set of $n$-tuples $a = (a_b)_{b \in B}$ of $A$-elements indexed by $B$, which can also be seen as functions from $B$ to $A$. For $T \subset B$, we also write $a_{|T} \mydef (a_t)_{t \in T}$ the restriction of the tuple $a$ to the coordinates of $T$.

\subsection{Two definitions for PIR protocols}

A vast majority of existing PIR schemes start by simply cloning the database $D$ on all the servers $S_1,\dots,S_\ell$. Then, the role of each server $S_j$ is to compute some combination of symbols from $D$, related to the query sent by $U$. This computation has a non-trivial cost, so in a certain sense, the computational complexity of the privacy of the  PIR scheme is mainly devoted to the servers.

More formally, one can define \emph{replication-based PIR protocols} as follows:
\begin{definition}[standard, or replication-based PIR protocol]
  Assume that every server $S_j$, $1 \le j \le \ell$, stores a copy of the database $D$. An $\ell$-server replication-based PIR protocol is a set of three algorithms $(\calQ, \calA, \calR)$ running the following steps on input $i \in [1,k]$:
  \begin{enumerate}
  \item \emph{Query generation:} the randomized algorithm $\calQ$ generates $\ell$ queries $(q_1, \hdots, q_{\ell}) \mydef \calQ(i)$. Query $q_j$ is sent to server $S_j$.
  \item \emph{Servers' answer:} each server $S_j$ computes an answer $a_j = \calA(q_j, D)$ and sends it back to the user\footnote{\label{note1}algorithm $\calA \mydef \calA_j$ may depend on $j$}.
  \item \emph{Reconstruction:} denote by $\mathbf{a} = (a_1,\dots,a_\ell)$ and $\mathbf{q} = (q_1,\dots,q_\ell)$. User $U$ computes and outputs $r = \calR(i, \mathbf{a}, \mathbf{q})$.
  \end{enumerate}
  The PIR protocol is said:
  \begin{itemize}
    \item \emph{correct} if $r = D_i$ when the servers follow the protocol;
    \item \emph{$t$-private} if, for every $(i, i') \in [1,k]^2$ and every $T \subseteq [1, \ell]$ such that $|T| \le t$, the distributions $\calQ(i)_{|T}$ and $\calQ(i')_{|T}$ are the same. We also say that the PIR protocol resists $t$ collusions of servers.
  \end{itemize}
  We call \emph{communication complexity} the number of bits sent between the user and the servers, and \emph{server} (resp. \emph{user}) \emph{computational complexity} the maximal number of $\FF_q$-operations made by a server in order to compute an answer $a_j$ (resp. made by $\calR$ to reconstruct the desired item).
\end{definition}

According to this definition, one sees that the servers must jointly carry the $\ell$ copies of the database, so the \emph{storage overhead} of the scheme is $(\ell-1) |D|$ bits. Moreover, since $D$ is a raw database without specific structure, the algorithm $\calA$ has no reason to be trivial and can incur superlinear computations for the servers  --- which is verified for most of current replication-based PIR protocols.

A way to reduce the computation cost of PIR protocols is to preprocess the database. Therefore we need to model PIR protocols for which the database can be encoded and distributed over the servers. From now on, let $c = (c_i)_{i \in I}$ denote \emph{an encoding of the database} $D$, \emph{i.e.} the image of $D$ by an injective map $\FF_q^k \to \FF_q^I$, with $|I| = n \ge k$. Besides, for convenience we assume that $I = [1,s] \times [1,\ell]$ and for readability we write $c_{(i_1, i_2)} = c_{i_1}^{(i_2)}$ and $c^{(j)} = (c_r^{(j)})_{r \in [1,s]}$.
\begin{definition}[distributed PIR protocol]
  \label{def:distributed-PIR}
  Assume that for $1 \le j \le \ell$, server $S_j$ holds the part $c^{(j)}$ of the encoded database. An $\ell$--server distributed PIR protocol is a set of three algorithms $(\calQ, \calA, \calR)$ running the following steps on input $i \in I$:
  \begin{enumerate}
  \item \emph{Query generation:} the randomized algorithm $\calQ$ generates $\ell$ queries $(q_1, \hdots, q_{\ell}) \mydef \calQ(i)$. Query $q_j$ is sent to server $S_j$.
  \item \emph{Servers' answer:} each server $S_j$ computes an answer $a_j = \calA(q_j, c^{(j)})$ and sends it back to the user.
  \item \emph{Reconstruction:}  denote by $\mathbf{a} = (a_1,\dots,a_\ell)$ and $\mathbf{q} = (q_1,\dots,q_\ell)$. User $U$ computes and outputs $r = \calR(i, \mathbf{a}, \mathbf{q})$.
  \end{enumerate}
\emph{Correctness} and \emph{privacy} properties are identical to those of replication-based PIR protocols. Similarly, one can also define \emph{communication} and \emph{computational complexities}, and since the database $D$ has been encoded, we finally define the \emph{storage overhead} as the number of redundancy bits stored by the servers, that is, $(s \ell - k)\log q$.
\end{definition}

In this paper, we focus on distributed PIR protocols with low computational complexity on the server side. More precisely, we build PIR protocols where the answering algorithm $\calA$ consists only in \emph{reading some symbols} of the database. Thus, our PIR protocols are computationally optimal on the server side, in a sense that, compared to the non-private retrieval, they incur no extra computational burden for the each server taken individually.

\subsection{PIR protocols from locally decodable codes}

As pointed out in the introduction, Augot \emph{et al.} \cite{AugotLS14} used a family of locally decodable codes (LDC) to design a distributed PIR scheme. LDCs are known to give rise to PIR protocols for a long time~\cite{KatzT00}, but we emphasize that the main idea from~\cite{AugotLS14} is to benefit from the fact that the encoded database can be smartly partitioned with respect to the queries of the local decoder.

Based on the seminal work of Katz and Trevisan~\cite{KatzT00}, we briefly remind how to design a PIR protocol based on a perfectly smooth locally decodable code. First, let us define (linear) locally decodable codes.
\begin{definition}[locally decodable code]
  Let $\Sigma$ be a finite set, $2 \le \ell \le k \le n$ be integers, and $\delta, \epsilon \in [0,1]$. A code $\calC : \Sigma^k \to \FF_q^n$ is $(\ell, \delta, \epsilon)$--locally decodable if and only if there exists a randomized algorithm $\calD$ such that, for every input $i \in [1,k]$ we have:
  \begin{itemize}
  \item for all $m \in \Sigma^k$ and all $y \in \FF_q^n$, if $|\{ j \in [1,n], y_j \ne \calC(m)_j \}| \le \delta n$, then
    \[
    \mathbb{P}( \calD^{(y)}(i) = m_i ) \ge 1 - \epsilon\,,
    \]
    where the probability is taken over the internal randomness of $\calD$;
  \item $\calD$ reads at most $\ell$ symbols $y_{q_1}, \dots, y_{q_\ell}$ of $y$.
  \end{itemize}
  Notation $\calD^{(y)}$ refers to the fact that $\calD$ has oracle access to single symbols $y_{q_j}$ of the word $y$. The parameter $\ell$ is called the \emph{locality} of the code. Moreover, the code $\calC$ is said \emph{perfectly smooth} if on an arbitrary input $i$, each individual query of the decoder $\calD$ is uniformly distributed over the coordinates of the word $y$.
\end{definition}

Now let us say a user wants to use a PIR protocol on a database $D \in \Sigma^k$, and assume there exists a perfectly smooth locally decodable code $\calC \subset \FF_q^n$ of dimension $k$ and locality $\ell$. Figure~\ref{fig:LDC-to-encoded-PIR} presents a distributed PIR protocol based on $\calC$.

\begin{figure}[h!]
  \setlength{\fboxsep}{6pt}
  \centering
  \fbox{\parbox{0.93\columnwidth}{
      \textbf{1) Initialization step.} User $U$ encodes $D$ into a codeword $c' \in \calC$. Each server $S_1,\dots,,S_\ell$ holds a copy of $c'$. In the formalism of Definition~\ref{def:distributed-PIR}, it means that $c^{(j)} := c'$, for $j = 1,\dots,\ell$.

      \textbf{2) Retrieving step for symbol $D_i$.} Denote by $\calD$ a local decoding algorithm for $\calC$.

      \begin{enumerate}
      \item \emph{Queries generation:} user $U$ calls $\calD$ to generate at random a query $(q_1, \dots, q_{\ell})$ for decoding the symbol $D_i$. Query $q_j$ is sent to server $S_j$.
      \item \emph{Servers' answer:} each server $S_j$ \emph{reads} the encoded symbol $a_j \mydef c'_{q_j}$. Then $S_j$ sends $a_j$ to $U$.
      \item \emph{Reconstruction:} user $U$ collects the $\ell$ codeword symbols $(c'_{q_j})_{j \in [1,\ell]}$ and feeds the local decoding algorithm $\calD$ in order to retrieve $D_i$.
      \end{enumerate}
    }
  }
  \caption{\label{fig:LDC-to-encoded-PIR}A distributed PIR protocol based on a locally decodable code $\calC$.}
\end{figure}

The main drawback of these LDC-based PIR protocols is their storage overhead, since the $\ell$ servers must store $\ell n/k = \ell/R$ times more data than the raw database ($R \mydef k/n$ represents the \emph{information rate}, or \emph{rate}, of the code). This issue becomes especially crucial as building LDCs with small locality and high rate is highly non-trivial.

The idea of Augot, Levy-dit-Vehel and Shikfa~\cite{AugotLS14} for reducing the storage overhead is to benefit from a natural partition of the support of multiplicity codes~\cite{KoppartySY14}. Assume that each codeword $c \in \calC$ can be \emph{split} into $\ell$ disjoint parts $c^{(1)}, \dots, c^{(\ell)}$, such that each coordinate $q_j$ of any possible query $(q_1,\dots,q_\ell)$ of the PIR protocol corresponds to reading some symbols on $c^{(j)}$. By sending the part $c^{(j)}$ to server $S_j$, the PIR protocol of Figure~\ref{fig:LDC-to-encoded-PIR} can be improved in order to save storage. We devote Section~\ref{sec:comparison} to more explanation on this construction, as well as to a comparison with our schemes.

Finally, one can notice that the communication complexity of LDC-based PIR protocols depends on the locality of the code, while the smoothness of the code serves their privacy. We also point out two important remarks.
\begin{enumerate}
\item Assuming a noiseless transmission and \emph{honest-but-curious} servers (\emph{i.e.} they want to discover the index of the desired symbol but never give wrong answers), one \emph{does not need} a powerful local decoding algorithm. Indeed, it should be possible to reconstruct the desired symbol $D_i$ by local decoding only one erasure on the codeword. For instance, computing a single low-weight parity-check sum should be enough.
\item Smoothness is sufficient for $1$-privacy, but we need more structure for preventing collusions of servers.
\end{enumerate}

Coupled with the fact that we want to split the database over several servers, these remarks lead us to design other kinds of encoding, which answer as close as possible the needs of private information retrieval protocols. Our construction relies on combinatorial structures, namely \emph{transversal designs}, that we recall in the upcoming section.

\section{Transversal designs and codes}
\label{sec:designs}

Let us give here the definition of transversal designs and how to build linear codes upon them. We refer to~\cite{AssmusK92},~\cite{Stinson04} and~\cite{ColbournD06} for complementary details.

\begin{definition}[block design]
  A \emph{block design} is a pair $\calD = (X, \calB)$ where $X$ is a finite set of so-called \emph{points}, and $\calB$ is a set of non-empty subsets of $X$ called the \emph{blocks}.
\end{definition}

\begin{definition}[incidence matrix]
  Let $\calD = (X, \calB)$ be a block design. An \emph{incidence matrix} $M_{\calD}$ of $\calD$ is a matrix of size $|\calB| \times |X|$, whose $(i,j)-$entry, for $i \in \calB$ and $j \in X$, is:
  \[
  \left\{
  \begin{array}{ll}
    1 & \text{ if the block } i \text{ contains the point } j,\\
    0 & \text{ otherwise.}
  \end{array}
  \right.
  \]
  The \emph{$q$-rank} of $M_{\calD}$ is the rank of $M_{\calD}$ over the field $\FF_q$.
\end{definition}

For $B \subset X$, the \emph{incidence vector} $\mathds{1}_B \in \{0,1\}^X$ is the row vector whose $x$-th coordinate is $1$ if and only if $x \in B$. Let us notice that, given a design $\calD = (X, \calB)$, one can build $M_{\calD}$ by stacking incidence vectors of blocks $B \in \calB$.

Of course, any design admits many incidence matrices, depending on the way points and blocks are ordered. However, all these incidence matrices are equal up to some permutation of their rows and columns, and, in particular, they all have the same $q$-rank. Hence, we call \emph{$q$-rank of a design} the $q$-rank of any of its incidence matrices. Moreover, from now on we consider incidence matrices of designs up to an ordering of points and blocks, and we abusively refer to \emph{the} incidence matrix $M_{\calD}$ of a design $\calD$.

\begin{example}\label{ex:AGdesign}
  Let $\AA^2(\FF_3)$ be the affine plane over the finite field $\FF_3$, and $X$ be the set consisting of its $9$ points:
  \[
  \setlength\arraycolsep{1pt}
  \begin{array}{rl}
  X = \{ &(0,0), (0,1), (0,2), (1,0), (1,1), \\
  &(1,2), (2,0), (2,1), (2,2) \,\}\,.
  \end{array}
  \]
  We define the block set $\calB$ as the set of the $12$ affine lines of $\AA^2(\FF_3)$:
  \[
    \setlength\arraycolsep{1pt}
    \begin{array}{rll}
      \calB = \{&\{ (0,0), (0,1), (0,2) \}, &\{ (1,0), (1,1), (1,2) \}, \\
                &\{ (2,0), (2,1), (2,2) \}, &\{ (0,0), (1,1), (2,2) \}, \\
                &\{ (1,0), (2,1), (0,2) \}, &\{ (2,0), (0,1), (1,2) \}, \\
                &\{ (0,0), (2,1), (1,2) \}, &\{ (1,0), (0,1), (2,2) \}, \\
                &\{ (2,0), (1,2), (0,2) \}, &\{ (0,0), (1,0), (2,0) \}, \\
                &\{ (0,1), (1,1), (2,1) \}, &\{ (0,2), (1,2), (2,2) \}\,\}\,.
    \end{array}
  \]
  The pair $\calD = (X, \calB)$ is then a block design, and its associated $(12\times9)$--incidence matrix is
  \[
  M_{\calD} =
  \left(\begin{array}{rrrrrrrrr}
1 & 1 & 1 & 0 & 0 & 0 & 0 & 0 & 0 \\
1 & 0 & 0 & 1 & 0 & 0 & 1 & 0 & 0 \\
1 & 0 & 0 & 0 & 1 & 0 & 0 & 0 & 1 \\
1 & 0 & 0 & 0 & 0 & 1 & 0 & 1 & 0 \\
0 & 1 & 0 & 1 & 0 & 0 & 0 & 0 & 1 \\
0 & 1 & 0 & 0 & 1 & 0 & 0 & 1 & 0 \\
0 & 1 & 0 & 0 & 0 & 1 & 1 & 0 & 0 \\
0 & 0 & 1 & 1 & 0 & 0 & 0 & 1 & 0 \\
0 & 0 & 1 & 0 & 1 & 0 & 1 & 0 & 0 \\
0 & 0 & 1 & 0 & 0 & 1 & 0 & 0 & 1 \\
0 & 0 & 0 & 1 & 1 & 1 & 0 & 0 & 0 \\
0 & 0 & 0 & 0 & 0 & 0 & 1 & 1 & 1
  \end{array}\right).
  \]
  A computation shows that over the field $\FF_2$, matrix $M_{\calD}$ is full-rank, while over $\FF_3$, it has only rank $6$.
\end{example}

\begin{definition}[transversal design]
  Let $s, \ell \ge 2$ and $\lambda \ge 1$ be integers. A \emph{transversal design}, denoted $\TD_{\lambda}(\ell, s)$, is a block design $(X, \calB)$ equipped with a partition $\calG = \{ G_1,\dots,G_\ell \}$ of $X$ called the set of \emph{groups}, such that:
  \begin{itemize}
  \item $|X| = \ell s$;
  \item any group in $\calG$ has size $s$ and any block in $\calB$ has size $\ell$;
  \item any unordered pair of elements from $X$ is contained either in one group and no block or in no group and $\lambda$ blocks.
  \end{itemize}
  If $\lambda = 1$, we use the simpler notation $\TD(\ell, s)$.
\end{definition}

\begin{remark}
  A block cannot be secant to a group in more than one point, otherwise the third condition of the definition would be disproved. Moreover, since the block size equals the number of groups, any block must meet any group. Hence the following holds:
  \[
  \forall (B, G) \in \calB \times \calG,\, |B \cap G| = 1\,.
  \]
  The definition also implies there must lie exactly $\lambda s^2$ blocks in $\calB$.
\end{remark}

\begin{example}
  Let $\calD = (X, \calB)$ be the block design defined in Example~\ref{ex:AGdesign}. Define $\calG$ to be any set of $3$ parallel lines from $\calB$ which partitions the point set $X$. For instance, one can consider
  \[
    \setlength\arraycolsep{1pt}
    \begin{array}{rl}
      \calG = \{ & \{ (0,0), (0,1), (0,2) \}, \\
                 & \{ (1,0), (1,1), (1,2) \}, \\
                 & \{ (2,0), (2,1), (2,2) \} \, \}\,.
    \end{array}
  \]
  Then, $\calT = (X, \calB \setminus \calG, \calG)$ is a transversal design $\TD(3, 3)$. Indeed, $\calT$ is composed of $\ell s = 9$ points, $\ell = 3$ groups of size $s = 3$ and $s^2 = 9$ blocks of size $\ell = 3$ each. Moreover, in the affine plane every unordered pair of points belongs simultaneously to a unique line, which is represented in $\calT$ either by a group or by a block. More generally, for any prime power $q$, a transversal design $\TD(q,q)$ can be built with the affine plane $\AA^2(\FF_q)$. A generalisation of this construction will be given in Subsection~\ref{subsec:affine-geometry}.
\end{example}

A simple way to build linear codes from block designs is to associate a parity-check equation of the code to each incidence vector of a block of the design. We recall that the \emph{dual code} $\calC^\perp$ of a code $\calC \subseteq \FF_q^n$ is the linear vector space consisting of vectors $h \in \FF_q^n$ such that $\forall c \in \calC, \sum_{i=1}^n c_i h_i = 0$. 

\begin{definition}[code of a design]\label{def:code-of-a-design}
  Let $\FF_q$ be a finite field, $\calD = (X, \calB)$ be a block design and $M_{\calD}$ be its incidence matrix. The code $\Code_q(\calD)$ is the $\FF_q$-linear code of length $|X|$ admitting $M_{\calD}$ as a parity-check matrix.
\end{definition}

\begin{remark}
  The code $\Code_q(\calD)$ is uniquely defined up to a chosen order of the points $X$. For different orders, the arising codes remain permutation-equivalent. Also notice that the way blocks are ordered does not affect the code.
\end{remark}

For any design $\calD$, the dimension over $\FF_q$ of $\Code_q(\calD)$ equals $|X| - \rank_q(M_{\calD})$. Since $M_\calD$ has coefficients in $\{0, 1\}$, one must notice that $\rank_q(M_{\calD}) = \rank_p(M_{\calD})$, where $p$ is the characteristic of the field $\FF_q$.

\begin{remark}
  Standard literature (\emph{e.g.}~\cite{AssmusK92}) sometimes defines $\Code_q(\calD)$ (and not $\Code_q(\calD)^{\perp}$) to be the vector space \emph{generated} by the incidence matrix of the design. We favor this convention because $\mathrm{Code}_q(\calD)$ will serve to \emph{encode} the database in our PIR scheme.
\end{remark}

\begin{example}
  The design $\calD$ from Example~\ref{ex:AGdesign} gives rise to $\calC = \mathrm{Code}_3(\calD)$, a linear code over $\FF_3$, of length $9$ and dimension $3$. A full-rank generator matrix of $\calC$ is given by:
  \[
  G = \left(\begin{array}{rrrrrrrrr}
1 & 1 & 1 & 1 & 1 & 1 & 1 & 1 & 1 \\
0 & 0 & 0 & 1 & 1 & 1 & 2 & 2 & 2 \\
0 & 1 & 2 & 0 & 1 & 2 & 0 & 1 & 2 
  \end{array}\right).
  \]
  One may notice that this code is the generalized Reed-Muller code of degree $1$ and order $2$ over $\FF_3$, that is, the evaluation code of bivariate polynomials of total degree at most $1$ over the whole affine plane $\FF_3^2$.
\end{example}

\begin{definition}[systematic encoding]\label{def:systematic-encoding}
  Let $\calC \subseteq \FF_q^n$ be a linear code of dimension $k \le n$. A systematic encoding for $\calC$ is a one-to-one map $\phi : \FF_q^k \to \calC$, such that there exists an injective map $\sigma : [1,k] \to [1,n]$ satisfying:
\[
\forall m \in \FF_q^k, \forall i \in [1,k],\, m_i = \phi(m)_{\sigma(i)}\,.
\]
The set $\sigma([1,k]) \subseteq [1,n]$ is called an \emph{information set} of $\calC$.
\end{definition}

In other words, a systematic encoding allows to view the message $m$ as a subword of its associated codeword $\phi(m) \in \calC$. For instance, it is useful for retrieving $m$ from $c$ efficiently, when the codeword $c$ has not been corrupted. A systematic encoding exists for any code $\calC$, is not necessarily unique, and can be computed through a Gaussian elimination over any generator matrix of the code. Also notice that this computation can be tedious for large codes.

\section{$1$-private PIR protocols based on transversal designs}
\label{sec:no-collusion-PIR}

In this section we present our construction of PIR protocols relying on transversal designs. The idea is that the knowledge of one point of a block of a transversal design gives (almost) no information on the other points lying on this block. The code associated to such a design then transfers this property to the coordinates of codewords. Hence, we obtain a PIR protocol which can be proven $1$-private, that is, which ensures perfect privacy for non-communicating servers. Though this protocol cannot resist collusions, we will see in Section~\ref{sec:collusion-PIR} that a natural generalisation leads to $t$-private PIR protocols with $t > 1$.

Notice that both Fazeli \emph{et al.}'s work~\cite{FazeliVY15} and ours make use of codes in order to save storage in PIR protocols. Nevertheless, we emphasize that the constructions are very different, since Fazeli \emph{et al.} emulate a PIR protocol from an existing one while we build our PIR protocols from scratch.

\subsection{The transversal-design-based distributed PIR protocol}

Let $\calT$ be a transversal design $\TD(\ell, s)$ and $n = |X| = \ell s$. Denote by $\calC = \Code_q(\calT) \subseteq \FF_q^n$ the associated $\FF_q$-linear code, and let $k = \dim_{\FF_q} \calC$. Our PIR protocol is defined in Figure~\ref{fig:no-collusion-PIR}. We then summarize the steps of the construction in Figure~\ref{fig:summary}.

\begin{figure}[h!]
  \setlength{\fboxsep}{6pt}
  \centering
  \fbox{\parbox{0.93\columnwidth}{

      \textbf{Parameters:} $\calT = (X, \calB, \calG)$ is a $\mathrm{TD}_{\lambda}(\ell, s)$; $\calC = \mathrm{Code}_q(\calT)$ has length $n = \ell s$ and dimension $k$.
      \medskip

      \textbf{1) Initialization step.}

      \begin{enumerate}
      \item \emph{Encoding.} User $U$ computes a systematic encoding of the database $D \in \FF_q^k$, resulting in the codeword $c \in \calC$.
      \item \emph{Distribution.} Denote by $c^{(j)} = c_{|G_j}$ the symbols of $c$ whose support is the group $G_j \in \calG$. Each server $S_j$ receives $c^{(j)}$, for $1 \le j \le \ell$.
      \end{enumerate}

      \textbf{2) Retrieving step for symbol $c_i$ for $i \in X$.} Denote by $j^* \in [1, \ell]$ the index of the unique group $G_{j^*}$ which contains $i$ --- that is, $c_i = c_r^{(j^*)}$ for some $r \in [1,s]$. Also denote by $\calB^*$ the subset of blocks containing $i$. The three steps of the distributed PIR protocol are:

      \begin{enumerate}
      \item \emph{Queries generation.} $U$ picks uniformly at random a block $B \in \calB^*$. For $j \ne j^*$, user sends the unique index $q_j \in B \cap G_j$ to server $S_j$. Server $S_{j^*}$ receives a random query $q_{j^*}$ uniformly picked in $G_{j^*}$. To sum up ($\xleftarrow{\$}$ stands for \enquote{picked uniformly at random in}):
        \[
        \left\{
        \begin{array}{rcll}
          \calQ(i)_{j^*} &\xleftarrow{\$} &G_{j^*}, &\text{ for } j^* \text{ s.t. } i \in G_{j^*}\\
          B             &\xleftarrow{\$} &\calB^* &\multirow{2}{*}{$\text{ for } j \ne j^*$}\\
          \calQ(i)_j    &\leftarrow      &B \cap G_j, &
        \end{array}
        \right.
        \]
      \item \emph{Servers' answer.} Each server $S_j$ (including $S_{j^*}$) reads $a_j \mydef c_{q_j}$ and sends it back to the user. That is,
        \[
        \calA(q_j, c^{(j)}) = c_{q_j}\,.
        \]
      \item \emph{Reconstruction.} Denote by $\mathbf{a} = \{ a_1, \dots, a_\ell \}$ and $\mathbf{q} = \{ q_1, \dots, q_\ell \}$. User $U$ computes
        \[
        r = \calR (i, \mathbf{a}, \mathbf{q} ) \mydef - \sum_{j \ne j^*} a_j = - \sum_{j \ne j^*} c_{q_j}
        \]
         and outputs $r$.
      \end{enumerate}
  \vskip -0.2cm
  }}
  \caption{A $1$-private distributed PIR protocol based on the $\FF_q$-linear code defined by a transversal design\label{fig:no-collusion-PIR}.}
\end{figure}

\begin{figure}[h!]
  
  \fbox{\parbox{0.93\columnwidth}{
      \[
      \xymatrixcolsep{4pc}
      \entrymodifiers={+!!<0pt,\fontdimen22\textfont2>}
      \xymatrix{
        \text{Transversal design } \TD(\ell, s) \ar[d]^-{\text{incidence matrix}} \\
        \text{TD-based linear code } \Code_q(\TD(\ell, s)) \subseteq \FF_q^{\ell s} \ar[d]^-{\text{database encoding}} \\
        \text{Distributed PIR scheme} 
      }
      \]
    }}
  \caption{\label{fig:summary}Summary of the steps leading to the construction of a transversal-design-based PIR scheme.}
\end{figure}
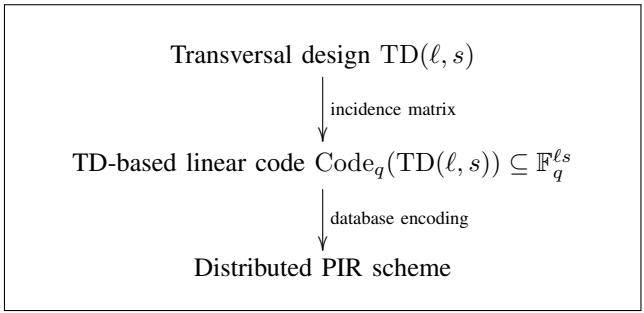

\subsection{Analysis}

We analyse our PIR scheme by proving the following:
\begin{theorem}\label{thm:1-private}
  Let $D$ be a database with $k$ entries over $\FF_q$, and $\calT = \mathrm{TD}(\ell, s)$ be a transversal design, whose incidence matrix has rank $\ell s - k$ over $\FF_q$. Then, there exists a distributed $\ell$-server $1$-private PIR protocol with:
  \begin{itemize}
  \item only one $\FF_q$-symbol to read for each server,
  \item $\ell-1$ field operations over $\FF_q$ for the user,
  \item $\ell \log(sq)$ bits of communication ($\ell \log s$ are uploaded, $\ell \log q$ are downloaded),
  \item a (total) storage overhead of $(\ell s - k) \log q$ bits on the servers.
  \end{itemize}
\end{theorem}
\begin{proof}
Recall the PIR protocol we are dealing with is defined in Figure~\ref{fig:no-collusion-PIR}.

\textbf{Correctness.} By definition of the code $\calC = \Code_q(\calT)$, the incidence vector $\mathds{1}_B$ of any block $B \in \calB$ belongs to the dual code $\calC^{\perp}$. Hence, for $c \in \calC$, the inner product $\mathds{1}_B \cdot c$ vanishes, or said differently, $\sum_{x \in B} c_x = 0$. We recall that $j^*$ represents the index of the group which contains $i$. Since the servers $S_j$, $j\ne j^*$, receive queries corresponding to the points of a block $B$ which contains $i$, we have $c_i = - \sum_{x \in B \setminus \{i\}} c_x = - \sum_{j \ne j^*} c_{q_j} $, and our PIR protocol is correct as long as there is no error on the symbols $a_j \mydef c_{q_j}$ returned by the servers.

\textbf{Security ($1$-privacy).} We need to prove that for all $j \in [1,\ell]$, it holds that $\PP(i \mid q_j) = \PP(i)$, where probabilities are taken over the randomness of $B \leftarrow \calB^*$. The law of total probability implies
\[
\begin{aligned}
  \PP(i \mid q_j) &= \PP(i \mid q_j\text{ and } i \in G_j) \,\PP(i \in G_j) \\
  &\quad\quad+ \PP( i \mid q_j \text{ and } i \notin G_j) \,\PP(i \notin G_j)\\
  &= \PP(i \mid i \in G_j) \,\PP(i \in G_j) \\
  &\quad\quad + \PP( i \mid i \notin G_j) \,\PP(i \notin G_j) \\
  &= \PP(i)\,,
\end{aligned}
\]
and the reasons why we eliminated the random variable $q_j$ in the conditional probabilities are:
\begin{itemize}
\item in the case $i \in G_j$ (that is, $j = j^*$), by definition of our PIR protocol we know that $q_j$ is uniformly random, so $q_j$ and $i$ are independent;
\item in the case $i \notin G_j$, by definition of a transversal design, there are as many blocks containing both $q_j$ and $i$ as there are blocks containing $q_j$ and any $i'$ in $X \setminus G_j$ (the number of such blocks is always $\lambda$). So once again, the value of the random variable $q_j$ is not related to $i$.
\end{itemize}

\textbf{Communication complexity.} Exactly one index in $[1,s]$ and one symbol in $\FF_q$ are exchanged between each server and the user. So the overall communication complexity is $\ell \times (\log(s) + \log(q)) = \ell \log(s q)$ bits.

\textbf{Storage overhead.} The number of bits stored on a server is $s \log q$, giving a total storage overhead of $(\ell s - k) \log q$, where $k = \dim \calC$.

\textbf{Computation complexity.} Each server $S_j$ only needs to read the symbol defined by query $q_j$, hence our protocol incurs no extra computational cost.
\end{proof}

Theorem~\ref{thm:1-private} shows that, if we want to optimize the practical parameters of our PIR scheme, we basically need to look for small values of $\ell$, the number of groups. However, one observes that the dimension $k$ of $\mathrm{Code}_q(\calT)$ strongly depends on $\ell$ and $n$, and tiny values of $\ell$ can lead to trivial or very small codes. This issue should be carefully taken into account, since instances with $k < \ell$ represent PIR protocols which are more communication expensive to use than the trivial one, which simply retrieves the whole database. Hence, it is very natural to raise the main issue of our construction:

\begin{problem}
  \label{prob:main}
  Find codes $\calC = \Code_q(\calT)$ arising from transversal designs $\calT = \TD(\ell, s)$ with few groups (small $\ell$) and large dimension $k = \dim_{\FF_q} \calC$ compared to their length $n = \ell s$. 
\end{problem}

We first give a negative result, stating that the characteristic of the field $\FF_q$ should be chosen very carefully in order to obtain non-trivial codes.

\begin{proposition}
  \label{prop:low-dim-td}
  Let $\calT = (X, \calB, \calG)$ be a $\TD_\lambda(\ell, s)$. Let $q = p^e$, $p$ prime. If $p \nmid \lambda s$, then \[
\Code_q(\calT) \subseteq \{ c \in \FF_q^{s \ell}, \forall G \in \calG, c_{|G} \in \mathrm{Rep}(s) \}\,,
\]
where $\mathrm{Rep}(s)$ represents the repetition code of length $s$. In particular, if $p \nmid \lambda s$, then $\Code_q(\calT)$ has dimension at most $\ell$.
\end{proposition}
\begin{proof}
  For $x \in X$, recall that $\calB_x = \{ B \in \calB, x \in B\}$, and denote by $a^{(x)} = \sum_{B \in \calB_x} \mathds{1}_B$. We know that $a^{(x)} \in \Code_q(\calT)^\perp$, since $\Code_q(\calT)^\perp$ is generated by $\{ \mathds{1}_B, B \in \calB \}$. Denote by $G_x \in \calG$ the only group that contains $x$. We see that:
\[
\left\{
\begin{array}{ll}
a^{(x)}_x = \lambda s &~ \\
a^{(x)}_i = 0 & \text{ for all } i \in G_x \setminus \{ x \} \\
a^{(x)}_j = \lambda & \text{ for all } j \in X \setminus G_x\,.
\end{array}
\right.
\]
Therefore $a^{(x)} - a^{(y)} = \lambda s(\mathds{1}_{\{x\}} - \mathds{1}_{\{y\}})$ if $x$ and $y$ lie in the same group $G$. If $p \nmid \lambda s$, then we get $\mathds{1}_{\{x\}} - \mathds{1}_{\{y\}} \in \Code_q(\calT)^\perp$. Let now 
\[
\calC = \mathrm{Span}_{\FF_q} \{ \mathds{1}_{\{x\}} - \mathds{1}_{\{y\}}, \forall x,y \in X \text{ s.t. } \{x, y\} \subset G \in \calG \}
\]
We see that $\calC^\perp = \{ c \in \FF_q^{s \ell}, \forall G \in \calG, c_{|G} \in \mathrm{Rep}(s) \}$. Therefore we obtain the expected result.
\end{proof}

In the perspective of Problem~\ref{prob:main}, the following section is devoted to the construction of transversal designs with high rate.

\section{Explicit constructions of $1$-private $\mathrm{TD}$-based PIR protocols}
\label{sec:explicit-constructions}

From now on, we denote by $\ell(k)$ the number of servers involved in a given PIR protocol running on a database with $k$ entries, and by $n(k)$ the actual number of symbols stored by all the servers. As it is proved in Theorem~\ref{thm:1-private}, these two parameters are crucial for the practicality of our PIR schemes, and they respectively correspond to the block size and the number of points of the transversal design used in the construction. In practice, we look for small values of $\ell$ and $n$ as explained in Problem~\ref{prob:main}.

In this section, we first give two classical instances of transversal designs derived from finite geometries (Subsections~\ref{subsec:affine-geometry} and~\ref{subsec:projective-geometry}), leading to good PIR parameters. We then show how \emph{orthogonal arrays} produce transversal designs, and we more deeply study a family of such arrays leading to high-rate codes. Subsection~\ref{subsec:divisible-codes} is finally devoted to another family of orthogonal arrays whose \emph{divisibility} properties ensure to give an upper bound on the storage overhead of related PIR protocols.

\subsection{Transversal designs from affine geometries}
\label{subsec:affine-geometry}

Transversal designs can be built with incidence properties between subspaces of an affine space.

\begin{construction}[Affine transversal design]
  \label{cons:td-affine}
  Let $\AA^m(\FF_q)$ be the affine space of dimension $m$ over $\FF_q$, and $H = \{H_1,\dots,H_q\}$ be $q$ hyperplanes that partition $\AA^m(\FF_q)$. We define a transversal design $\calT_A(m,q)$ as follows:
\begin{itemize}
\item the point set $X$ consists in all the points in $\AA^m(\FF_q)$;
\item the groups in $\calG$ are the $q$ hyperplanes from $H$;
\item the blocks in $\calB$ are all the $1$-dimensional affine subspaces (lines) which do not entirely lie in one of the $H_j$, $j\in [1,q]$. We also say that such lines are \emph{secant} to the hyperplanes in $H$.
\end{itemize}
\end{construction}

The design thus defined is a $\TD(q, q^{m-1})$, since an affine line is either contained in one of the $H_j$, or is $1$-secant (\emph{i.e.} has intersection of size $1$) to each of them. To complete the study of the parameters of the induced PIR protocol, it remains to compute the dimension of $\Code(\calT_A(m,q))$.

Proposition~\ref{prop:low-dim-td} first proves that if $p$ does not divide $\lambda s = q$, then the code $\Code_p(\calT_A(m,q))$ has poor dimension. Since our goal is to obtain the largest codes as possible, we choose $p$ to be, for instance, the characteristic of the field $\FF_q$.

Now notice that all blocks of $\calT_A(m,q)$ belong to the block set of the \emph{affine geometry design} $\AG_1(m, q)$ --- which is defined as the incidence structure of all points and affine lines in $\AA^m(\FF_q)$. Thus, the incidence matrix $M_{\calT_A(m,q)}$ is a sub-matrix of $M_{\AG_1(m, q)}$, which implies that $\Code_p(\AG_1(m, q)) \subseteq \Code_p(\calT_A(m,q))$ for any field $\FF_p$. In fact, equality holds as shows the following result.

\begin{proposition}
  \label{prop:equality-td-ag}
  For every $q = p^e$ and $m \ge 2$, we have
  \[
  \Code_p(\AG_1(m, q)) = \Code_p(\calT_A(m, q))\,.
  \]
\end{proposition}
\begin{proof}
  Denote by $\calB^{(\AG)}$ the blocks of $\AG_1(m, q)$, and by $\calB^{(\calT)}$ and $\calG^{(\calT)}$ the blocks and groups of $\calT_A(m, q)$. Thanks to the previous discussion, we only need to show that for every block $B \in \calB^{(\AG)}$ contained in a group $G \in \calG^{(\calT)}$, it holds that $\mathds{1}_B \in \Code_p(\calT_A(m, q))^\perp$. For this sake, first notice that $\Code_p(\calT_A(m, q))^\perp = \mathrm{Span} \{ \mathds{1}_{B'}, B' \in \calB^{(\calT)} \}$.

Let now $G \in \calG^{(\calT)}$ and $B \in \calB^{(\AG)}$ such that $B \subseteq G$. Recall that $G$ is a hyperplane of $\AA^m(\FF_q)$, and let $P$ be a $2$-dimensional affine plane of $\AA^m(\FF_q)$ such that $P \cap G = B$. We claim that $\mathds{1}_P \in \mathrm{Span} \{ \mathds{1}_{B'}, B' \in \calB^{(\calT)} \}$. Indeed, $P$ admits a partition into affine lines which are secant to every hyperplane in $\calG$. Thus $\mathds{1}_P$ can be written as sum of the characteristic vectors of these lines.

Now let $x \in B$, and $\calB^{(\calT)}_{x, P} \mydef \{ B' \in \calB^{(\calT)},\, x \in B' \subset P \} \subseteq \calB^{(\calT)}$. Define $b^{(x)} = \sum_{B' \in \calB^{(\calT)}_{x, P}} \mathds{1}_{B'}$. It is clear that $b^{(x)} \in \mathrm{Span} \{ \mathds{1}_{B'}, B' \in \calB^{(\calT)} \}$, and we can notice that
\[
\left\{
\begin{array}{ll}
b^{(x)}_x = q = 0, &~ \\
b^{(x)}_i = 0 & \text{ for all } i \in B \setminus \{ x \}, \\
b^{(x)}_j = 1 & \text{ for all } j \in P \setminus B\,.
\end{array}
\right.
\]
In other words, $b^{(x)} = \mathds{1}_P - \mathds{1}_B$, therefore $\mathds{1}_B \in \mathrm{Span} \{ \mathds{1}_{B'}, B' \in \calB^{(\calT)} \}$.
\end{proof}

The benefit to consider $\AG_1(m, q)$ is that the $p$-rank of its incidence matrix has been well-studied. For instance, Hamada~\cite{Hamada68} gives a generic formula to compute the $p$-rank of a design coming from projective geometry. Yet, as presented in Appendix~\ref{app:hamada}, asymptotics are hard to derive from his formula for a generic value of $m$.

However, if $m = 2$, we know that $\rank_p (\AG_1(2, p^e)) = \binom{p+1}{2}^e$, which implies that 
\[
\dim(\Code_p(\calT_A(2,p^e))) = p^{2e} - \textstyle\binom{p+1}{2}^e\,.
\]
Hence we obtain the following family of PIR protocols.
\begin{proposition}
  Let $D$ be a database with $k = p^{2e} - \binom{p+1}{2}^e$ entries, $p$ a prime, $e \ge 1$. There exists a distributed $1$-private PIR protocol for $D$ with:
  \[
  \ell(k) = p^e \quad \text{ and } \quad n(k) = p^{2e}\,.
  \]
  For fixed $p$ and $k \to \infty$, we have
  \begin{equation}
    \label{eq:asymptotics-prop}
    \begin{array}{rl}
      \ell(k) &= \sqrt{k} + \Theta(k^{\frac{1}{2}+c_p}) \quad \text{ and } \\
      n(k)/k &= \frac{1}{1 - \left(\frac{1+1/p}{2}\right)^e} = 1 + \Theta(k^{c_p}) \to 1\,,
    \end{array}
  \end{equation}
  where $c_p = \frac{1}{2} \log_p(\frac{1+1/p}{2}) < 0$.
\end{proposition}

\begin{proof}
  The existence of the PIR protocol is a consequence of the previous discussion, using the family of codes $\Code_p(\calT_A(2,p^e))$. Let us state the asymptotics of the parameters. Recall we fix the prime $p$ and we let $e \to \infty$. First we have:
\begin{equation}
  \label{eq:asymptotics-1}
  \begin{aligned}
    n(k)/k &= \frac{p^{2e}}{p^{2e} - \binom{p+1}{2}^e} = \frac{1}{1 - \left(\frac{1+1/p}{2}\right)^e}\\
    &= 1 + \left(\frac{1+1/p}{2}\right)^e + \calO\left(\left(\frac{1+1/p}{2}\right)^{2e}\right)\,.
  \end{aligned}
\end{equation}
Notice that 
\[
\begin{aligned}
  \log_p k &= 2e + \log_p\left(1 - \left(\frac{1+1/p}{2}\right)^e\right) \\
  &= 2e + \calO\left(\left(\frac{1+1/p}{2}\right)^{e}\right)\,.
\end{aligned}
\]
Hence,
\[
\begin{aligned}
  \left(\frac{1+1/p}{2}\right)^e &= \left(\frac{1+1/p}{2}\right)^{\frac{1}{2}\log_p k + \calO\left(\left(\frac{1+1/p}{2}\right)^{e}\right)} \\
  &= k^{\frac{1}{2}\log_p\left(\frac{1 + 1/p}{2}\right)} \times \left(\frac{1+1/p}{2}\right)^{\calO\left(\left(\frac{1+1/p}{2}\right)^{e}\right)} \\
  &= \Theta\left(k^{c_p}\right)\,,
\end{aligned}
\]
since $\left(\frac{1+1/p}{2}\right)^{\calO\left(\left(\frac{1+1/p}{2}\right)^{e}\right)} \to 1$. Using \eqref{eq:asymptotics-1} we obtain the asymptotics we claimed on $n(k)/k$.

For $\ell(k)$, we see that $n(k) = \ell(k)^2$. Therefore, we get 
\[
\ell(k) = \sqrt{k} \sqrt{n(k)/k} = \sqrt{k} \sqrt{1 + \Theta(k^{c_p})} = \sqrt{k} + \Theta(k^{\frac{1}{2}+c_p})\,. 
\]
\end{proof}

We give in Table~\ref{tab:values1} the dimension of some codes arising from affine transversal designs. Notice that $m$ is not restricted to $2$, but we focus on codes with large, since they aimed at being applied in PIR protocols.

\begin{table}[h!]
  \[
  \begin{array}{c|c|c|c|c}
    m &    \ell = q &   n = s \ell = q^m & k = \dim \mathcal{C}  & R = k / n \\
    \hline
    2 &           8 &                 64 &                 37  & 0.578 \\
    2 &          16 &                256 &                175  & 0.684 \\
    2 &          32 &               1024 &                781  & 0.763 \\
    2 &          64 &         \num{4096} &         \num{3367}  & 0.822 \\
    2 &  \num{1024} &      \num{1048576} &       \num{989527}  & 0.944   \\
    2 &  \num{4096} &     \num{16777216} &     \num{16245775}  & 0.968   \\
    2 & \num{16384} &    \num{268435456} &    \num{263652487}  & 0.982   \\
    2 & \num{65536} &   \num{4294967296} &   \num{4251920575}  & 0.990   \\
    \hline
    3 &           8 &                512 &                139  & 0.271 \\
    3 &          16 &         \num{4096} &               1377  & 0.336 \\
    3 &          64 &       \num{262144} &       \num{118873}  & 0.453   \\
    3 &         256 &     \num{16777216} &      \num{9263777}  & 0.552   \\
    3 &  \num{1024} &   \num{1073741824} &    \num{680200873}  & 0.633   \\
    3 &  \num{8192} & \num{549755813888} & \num{400637408211}  & 0.729   \\    
    \hline
    4 &           8 &         \num{4096} &                406  & 0.099   \\
    4 &          64 &     \num{16777216} &      \num{2717766}  & 0.162   \\
    4 &         256 &   \num{4294967296} &    \num{890445921}  & 0.207   \\
    \hline
    5 &           8 &        \num{32768} &                994  & 0.030   \\
    5 &          64 &   \num{1073741824} &     \num{44281594}  & 0.041   \\
  \end{array}
  \]
  \caption{Dimension and rate of binary codes $\calC$ arising from $\calT_A(m,q)$. Remind that the rate $R$ of the code is related to the server storage overhead of the PIR protocol, and that $q = \ell$ is essentially the communication complexity and the number of servers.\label{tab:values1}}
\end{table}

Finally, for a better understanding of the parameters we can point out two PIR instances:
\begin{itemize}
\item choosing $m=2$ and $\ell = 4096$, there exists a PIR protocol on a $\simeq 2.0$ MB file with $6$ kB of communication and only $3.2 \%$ storage overhead;
\item for a $\simeq 46$ GB database ($m = 3$, $\ell = 8192$), we obtain a PIR protocol with $39$ kB of communication and $27 \%$ storage overhead.
\end{itemize}

\subsection{Transversal designs from projective geometries}
\label{subsec:projective-geometry}

The projective space $\PP^m(\FF_q)$ is defined as $(\AA^{m+1}(\FF_q) \setminus \{ {\bf 0} \}) / \sim\,$, where for $({\bf P}, {\bf Q}) \in (\AA^{m+1}(\FF_q) \setminus \{ {\bf 0} \})^2$, we have ${\bf P} \sim {\bf Q}$ if and only if there exists $\lambda \in \FF_q$ such that ${\bf P} = \lambda {\bf Q}$. A projective subspace can be defined as the zero set of a collection of linear forms over $\FF_q^{m+1}$. In particular, a projective hyperplane is the zero-set of one non-zero linear form over $\FF_q^{m+1}$.

Projective geometries are closely related to affine geometries, but contrary to them, there is no partition of the projective space into hyperplanes, since every pair of distinct projective hyperplanes intersects in a projective space of co-dimension $2$. To tackle this problem, an idea is to consider the hyperplanes $H_i$ which intersect on a fixed subspace of co-dimension $2$ (call it $\Pi_{\infty}$). Then, all the sets $H_i \setminus \Pi_{\infty}$ are disjoint, and their union gives exactly $\PP^m(\FF_q) \setminus \Pi_{\infty}$, where $\PP^m(\FF_q)$ denotes the projective space of dimension $m$ over $\FF_q$. Besides, any projective line disjoint from $\Pi_{\infty}$ is either contained in one of the $H_i$, or is $1$-secant to all of them. It results to the following construction:

\begin{construction}[Projective transversal design]
  \label{cons:td-projective}
  Let $\PP^m(\FF_q)$ and $\Pi_\infty$ defined as above. Let us define
\begin{itemize}
\item a point set $X = \PP^m(\FF_q) \setminus \Pi_{\infty}$;
\item a group set $\calG = \{ \text{projective hyperplanes } H \subset \PP^m(\FF_q),\, \Pi_{\infty} \subset H \}$;
\item a block set $\calB = \{ \text{projective lines } L \subset \PP^m(\FF_q),\, L \cap \Pi_{\infty} = \varnothing \text{ and } \forall H \in \calG,\, L \not\subset H \}$\,.
\end{itemize}
Finally, denote by $\calT_P(m, q) \mydef ( X, \calB, \calG )$.
\end{construction}

The design $\calT_P(m, q)$ is a $\TD(q+1, q^{m-1})$ and, as in the affine setting, its $p$-rank is related to that of $\PG_1(m, q)$, the classical design of point-line incidences in the projective space $\PP^m(\FF_q)$. Indeed, the incidence matrix $M$ of $\calT_P(m, q)$ is a submatrix of $M_{\PG_1(m, q)}$ from which we removed:
\begin{itemize}
\item the columns corresponding to the points in $\Pi_{\infty}$,
\item the rows corresponding to the lines not in $\calB$.
\end{itemize}
Said differently, the code associated to $\calT_P(m, q)$ contains (as a subcode) the $\Pi_{\infty}$-shortening of the code associated to $\PG_1(m, q)$. Hence $\dim \Code(\calT_P(m, q)) \ge \dim \Code(\PG_1(m, q)) - |\Pi_{\infty}|$. Contrary to Proposition~\ref{prop:equality-td-ag}, we could not prove equality, but this is of little consequence: up to using a subcode of $\Code(\calT_P(m, q))$ we can consider PIR protocols on databases with $k$ entries, where $k = \dim \Code(\PG_1(m, q)) - |\Pi_{\infty}|$.

Once again, for projective geometries Hamada's formula gets simpler for $m=2$, and leads to the following proposition. 

\begin{proposition}
  Let $D$ be a database with $k = p^{2e} + p^e - \binom{p+1}{2}^e - 1$ entries, $p$ a prime and $e \ge 1$. There exists a distributed $1$-private PIR protocol for $D$ with:
  \[
  \ell(k) = p^e + 1 \quad \text{ and } \quad n(k) = p^{2e} + p^e\,.
  \]
  Asymptotics are the same as in Equation~\eqref{eq:asymptotics-prop}.
\end{proposition}

In order to emphasize that the two previous constructions are asymptotically the same, we draw the rates of the codes involved in these two kinds of PIR schemes in Figure~\ref{fig:rates-projective-affine}.

\begin{figure}[h!]
  \centering
  \includegraphics{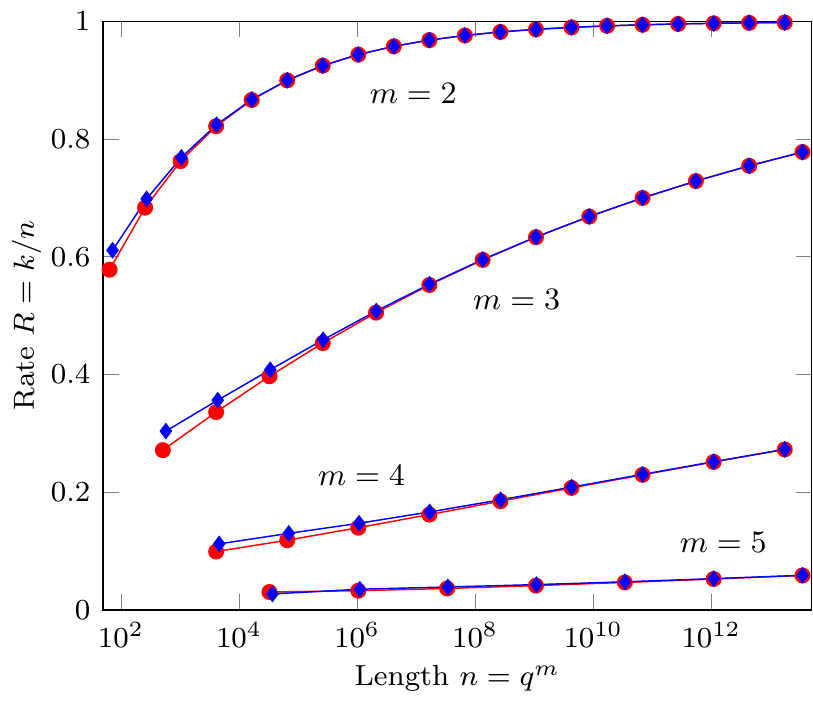}

  % \begin{tikzpicture}[scale=1.05]
  %   \begin{axis}[
  %     xmode=log, xmin=50, xmax=50000000000000, ymin=0, ymax=1,
  %     yticklabel style = {font=\footnotesize},
  %     xticklabel style = {font=\footnotesize},
  %     ylabel style = {at={(0.11,0.5)},anchor=south,font=\footnotesize},
  %     xlabel style = {at={(0.5,0.02)},anchor=north,font=\footnotesize},
  %     xlabel={Length $n = q^m$}, ylabel={Rate $R = k/n$}
  %     ]
  %     \addplot [mark=*, color=red] table {../data/rates_aff_m2.dat};
  %     \addplot [mark=*, color=red] table {../data/rates_aff_m3.dat};
  %     \addplot [mark=*, color=red] table {../data/rates_aff_m4.dat};
  %     \addplot [mark=*, color=red] table {../data/rates_aff_m5.dat};

  %     \addplot [mark=diamond*, color=blue] table {../data/rates_pro_m2.dat};
  %     \addplot [mark=diamond*, color=blue] table {../data/rates_pro_m3.dat};
  %     \addplot [mark=diamond*, color=blue] table {../data/rates_pro_m4.dat};
  %     \addplot [mark=diamond*, color=blue] table {../data/rates_pro_m5.dat};
  %   \end{axis}
  %   \draw (3, 5) node {\small $m=2$};
  %   \draw (4, 3) node {\small $m=3$};
  %   \draw (2.5, 1.3) node {\small $m=4$};
  %   \draw (6, 0.65) node {\small $m=5$};
  % \end{tikzpicture}

  \caption{\label{fig:rates-projective-affine}Rate of binary codes coming from $\calT_A(m, q)$ (in red) and $\calT_P(m, q)$ (in blue). For every fixed $m$, we let $q$ grow.}
\end{figure}

\subsection{Orthogonal arrays and the \emph{incidence code} construction}
\label{subsec:oa-to-td}

In this subsection, we first recall a way to produce plenty of transversal designs from other combinatorial constructions called \emph{orthogonal arrays}. 

\begin{definition}[orthogonal array]
  \label{def:oa}
  Let $\lambda, s \ge 1$ and $\ell \ge t \ge 1$, and let $A$ be an array with $\ell$ columns and $\lambda s^t$ rows, whose entries are elements of a set $S$ of size $s$. We say that $A$ is an orthogonal array $\OA_{\lambda}(t, \ell, s)$ if, in any subarray $A'$ of $A$ formed by $t$ columns and all its rows, every row vector from $S^t$ appears exactly $\lambda$ times in the rows of $A'$. We call $\lambda$ the \emph{index} of the orthogonal array, $t$ its \emph{strength} and $\ell$ its \emph{degree}. If $t$ (resp. $\lambda$) is omitted, it is understood to be $2$ (resp. $1$). If both these parameters are omitted we write $A = \OA(\ell, s)$. 
\end{definition}

From now on, for convenience we restrict Definition~\ref{def:oa} to orthogonal arrays with no repeated column and no repeated row. Next paragraph introduces a link between orthogonal arrays and transversal designs.

\subsubsection{Construction of transversal designs from orthogonal arrays} We can build a transversal design $\TD(\ell, s)$ from an orthogonal array $\OA(\ell, s)$ with the following construction, given as a remark in~\cite[ch.II.2]{ColbournD06}.

\begin{construction}[Transversal designs from orthogonal arrays]
  \label{cons:td-from-oa}
Let $A$ be an $\OA(\ell, s)$ of strength $t=2$ and index $\lambda=1$ with symbol set $S$, $|S| = s$, and denote by $\mathrm{Rows}(A)$ the $s^2$ rows of $A$. We define the point set $X = S \times [1, \ell]$. To each row $c \in \mathrm{Rows}(A)$ we associate a block
\[
B_c \mydef \{ (c_i, i), i \in [1,\ell] \}\,,
\]
so that the block set is defined as
\[
\calB \mydef \{ B_c, c \in \mathrm{Rows}(A) \}\,.
\]
Finally, let $\calG \mydef \{ S \times \{i\}, i \in [1, \ell] \}$. Then $(X, \calB, \calG)$ is a transversal design $\TD(\ell, s)$.
\end{construction}

\begin{example}

A very simple example of this construction is given in Figure~\ref{fig:OA-to-TD}, where for clarity we use letters for elements of the symbol set $\{a,b\}$, while the columns are indexed by integers. On the left-hand side, $A$ is an $\OA_1(2,3,2)$ with symbol set $\{a, b\}$. On the right-hand side, the associated transversal design $\TD(3, 2)$ is represented as a hypergraph: the nodes are the points of the design, the \enquote{columns} of the graph form the groups, and a block consists in all nodes linked with a path of a fixed color. One can check that every pair of nodes either belongs to the same group or is linked with one path.

\begin{figure}[h!]
  \centering
  $A = 
    \left[\begin{array}{ccc}
          \rowcolor{blue!20} a & b & b \\
          \rowcolor{green!20} b & b & a \\
          \rowcolor{black!20} b & a & b \\
          \rowcolor{red!20} a & a & a 
      \end{array}\right]
  \,\, \Longrightarrow \,\,
  $
  \begin{tikzpicture}[baseline=-0.6cm, scale=0.5]
    \node (A) at (0,0) {$(a,1)$};
    \node (B) at (3,0) {$(a,2)$};
    \node (C) at (6,0) {$(a,3)$};
    \node (D) at (0,-2) {$(b,1)$};
    \node (E) at (3,-2) {$(b,2)$};
    \node (F) at (6,-2) {$(b,3)$};

    \draw[red, very thick] (A) -- (B) -- (C);
    \draw[blue, very thick] (A) -- (E) -- (F);
    \draw[green!70!black, very thick] (D) -- (E) -- (C);
    \draw[black, very thick] (D) -- (B) -- (F);
    
  \end{tikzpicture}
  
  \caption{\label{fig:OA-to-TD}A representation of the construction of a transversal design from an orthogonal array.}
\end{figure}
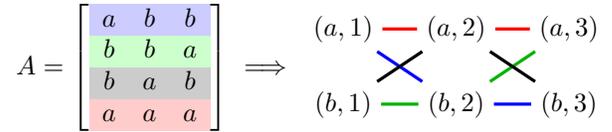

\end{example}

\begin{remark}\label{rem:code-OA}
   Listed in rows, all the codewords of a (generic) code $\calC_0$ give rise to an orthogonal array, whose strength $t$ is derived from the dual distance $d'$ of $\calC_0$ by $t = d'-1$. Notice that for linear codes, the dual distance is simply the minimum distance of the dual code, but it can also be defined for non-linear codes (see~\cite[Ch.5.\S5.]{MacWilliamsS77}). More details about the link between orthogonal arrays and codes can also be found in~\cite{ColbournD06}. For example, the orthogonal array of Figure~\ref{fig:OA-to-TD} comes from the binary parity-check code of length $3$ (by replacing $a$ by $0$ and $b$ by $1$). One can check that its dual distance is $3$ and its associated transversal design has strength $2$.
\end{remark}

Given a code $\calC_0$, we denote by $A_{\calC_0}$ the orthogonal array it defines (see Remark~\ref{rem:code-OA}) and by $\calT_{\calC_0}$ the transversal design built from $A_{\calC_0}$ thanks to Construction~\ref{cons:td-from-oa}.

\begin{example}\label{ex:TD-from-OA-from-RS}
  Let $\mathbf{x} = (x_1,\hdots,x_\ell)$ be an $\ell$-tuple of pairwise distinct elements of $\FF_q$ and denote by $\mathrm{RS}_2(\mathbf{x})$ the Reed-Solomon code of length $\ell$ and dimension $2$ over $\FF_q$ with evaluation points $\mathbf{x}$:
  \[
  \mathrm{RS}_2(\mathbf{x}) \mydef \{ (f(x_1),\hdots,f(x_\ell)),\, f \in \FF_q[X],\, \deg f < 2 \}\,.
  \]
  Then, $\mathrm{RS}_2(\mathbf{x})$ has dual distance $3$, so its codewords form an orthogonal array $A_{\mathrm{RS}_2(\mathbf{x})}  = \OA(\ell, q)$ of strength $2$. Now, one can use Construction~\ref{cons:td-from-oa} to obtain a transversal design $\calT_{\mathrm{RS}_2(\mathbf{x})} = \TD(\ell, q)$. The point set is $X = \FF_q \times [1, \ell]$, and the blocks are \enquote{labeled Reed-Solomon codewords}, that is, sets of the form $\{ (c_i, i), \, i \in [1, \ell] \}$ with $c \in  \mathrm{RS}_2(\mathbf{x})$. The $\ell$ groups correspond to the $\ell$ coordinates of the code: $G_i = \FF_q \times \{i\}$, $1 \le i \le \ell$.
\end{example}

We can finally sum up our construction by introducing the code $\Code_q(\calT_{\calC_0})$ arising from the transversal design defined by $\calC_0$. To the best of our knowledge, the construction $\calC_0 \mapsto \Code_q(\calT_{\calC_0})$ is new. We name $\Code_q(\calT_{\calC_0})$ the \emph{incidence code} of $\calC_0$, since its parity-check matrix $M_{\calT_{\calC_0}}$ essentially stores incidence relations between all the codewords in $\calC_0$. % We also motivate Our motivation in studying  the structure of $\calC_0$ may lead to properties of  $\Code_q(\calT_{\calC_0})$ which is used in PIR schemes.

\begin{definition}[incidence code]
  Let $\calC_0$ be a (generic) code of length $\ell$ over an alphabet $S$ of size $s$. The \emph{incidence code} of $\calC_0$ over $\FF_q$, denoted $\IC_q(\calC_0)$, is the $\FF_q$-linear code of length $n = s \ell$ built from the transversal design $\calT_{\calC_0}$, that is:
\[
\IC_q(\calC_0) := \Code(\calT_{\calC_0})\,.
\]
Notice that the field $\FF_q$ does not need to be the alphabet $S$ of the code $\calC_0$.
\end{definition}

Incidence codes are introduced in order to design PIR protocols, as summarizes Figure~\ref{fig:summary2}. We can show that, if $\calC_0$ has dual distance more than $3$, then the induced PIR protocol is $1$-private. A generalisation is formally proved in Corollary~\ref{coro:dual-distance-t-private}.

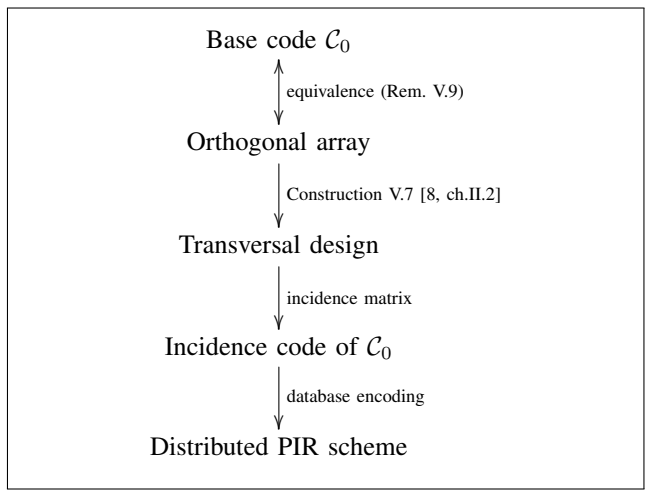
\begin{figure}[h!]
  \fbox{\parbox{0.93\columnwidth}{
      \[
      \xymatrixcolsep{4pc}
      \entrymodifiers={+!!<0pt,\fontdimen22\textfont2>}
      \xymatrix{
        \text{Base code $\calC_0$} \ar@{<->}[d]^-{\text{equivalence (Rem.~\ref{rem:code-OA})}} \\
        \text{Orthogonal array} \ar[d]^-{\text{Construction~\ref{cons:td-from-oa}~\cite[ch.II.2]{ColbournD06}}} \\
        \text{Transversal design} \ar[d]^-{\text{incidence matrix}} \\
        \text{Incidence code of $\calC_0$} \ar[d]^-{\text{database encoding}} \\
        \text{Distributed PIR scheme} \\
      }
      \]
    }}
  \caption{\label{fig:summary2}A distributed PIR scheme using the incidence code construction.}
\end{figure}

\begin{example}
  Here we provide a full example of the construction of an incidence code. Let $\calC_0$ be the full-length Reed-Solomon code of dimension $2$ over the field $\FF_4 = \{ 0, 1, \alpha, \alpha^2 = \alpha + 1 \}$. The orthogonal array associated to $\calC_0$ is composed by the following list of codewords:

  \[
  A =
  \setlength{\arraycolsep}{1pt}
  \begin{pmatrix}%{cccc}
    0,& 0,& 0,& 0\\
    1,& 1,& 1,& 1\\
    \alpha,& \alpha,& \alpha,& \alpha\\
    \alpha^2,& \alpha^2,& \alpha^2,& \alpha^2\\
    0,& 1,& \alpha,& \alpha^2\\
    0,& \alpha,& \alpha^2,& 1\\
    0,& \alpha^2,& 1,& \alpha\\
    1,& 0,& \alpha^2,& \alpha\\
    1,& \alpha^2,& \alpha,& 0\\
    1,& \alpha,& 0,& \alpha^2\\
    \alpha,& \alpha^2,& 0,& 1\\
    \alpha,& 0,& 1,& \alpha^2\\
    \alpha,& 1,& \alpha^2,& 0\\
    \alpha^2,& \alpha,& 1,& 0\\
    \alpha^2,& 1,& 0,& \alpha\\
    \alpha^2,& 0,& \alpha,& 1\\
  \end{pmatrix}
  \]
  Using Construction~\ref{cons:td-from-oa}, we get a transversal design $\calT_{\calC_0} = (X, \calB, \calG)$ with $16$ points ($4$ groups made of $4$ points) and $16$ blocks. Let us recall how we map a row of $A$ to a word in $\{0,1\}^{16}$. For instance, consider the fifth row:
\[
a := A_5 = (0, 1, \alpha, \alpha^2)\,.
\]
We turn $a$ into a block $B_a := \{ (0, 1), (1, 2), (\alpha, 3), (\alpha, 4) \} \in \calB$, and we build the incidence vector $\mathds{1}_{B_a}$ of the block $B_a$ over the point set $X = \{ (\beta, i), i \in [1,4], \beta \in \FF_4 \}$. Of course, in order to see $\mathds{1}_{B_a}$ as a word in  $\{0,1\}^{16}$, we need to order elements in $X$, for instance:
\[
\setlength{\arraycolsep}{2pt}
\begin{array}{rl}
  \big( &(0, 1), (1, 1), (\alpha, 1), (\alpha^2, 1), \\
        &(0, 2), (1, 2), (\alpha, 2), (\alpha^2, 2), \\
        &(0, 3), (1, 3), (\alpha, 3), (\alpha^2, 3), \\
        &(1, 4), (1, 4), (\alpha, 4), (\alpha^2, 4) \,\,\big)\,.
\end{array}
\]
Using this ordering, we get:
\[
\mathds{1}_{B_a} = \big( 1,0,0,0,0,1,0,0,0,0,1,0,0,0,0,1 \big) \in \{0,1\}^{16}\,.
\]

By computing all the $\mathds{1}_{B_a}$ for $a \in \mathrm{Rows}(A)$, we obtain the incidence matrix $M$ of the transversal design $\calT_{\calC_0}$:
  \[
  M = \left(
  \begin{array}{c:c:c:c}
    1000 & 1000 & 1000 & 1000 \\
    0100 & 0100 & 0100 & 0100 \\
    0010 & 0010 & 0010 & 0010 \\
    0001 & 0001 & 0001 & 0001 \\
    1000 & 0100 & 0010 & 0001 \\
    1000 & 0010 & 0001 & 0100 \\
    1000 & 0001 & 0100 & 0010 \\
    0100 & 1000 & 0001 & 0010 \\
    0100 & 0001 & 0010 & 1000 \\
    0100 & 0010 & 1000 & 0001 \\
    0010 & 0001 & 1000 & 0100 \\
    0010 & 1000 & 0100 & 0001 \\
    0010 & 0100 & 0001 & 1000 \\
    0001 & 0010 & 0100 & 1000 \\
    0001 & 0100 & 1000 & 0010 \\
    0001 & 1000 & 0010 & 0100 \\
  \end{array}\right)\,,
  \]
Notice that this matrix can be quickly obtained by respectively replacing entries $0$, $1$, $\alpha$ and $\alpha^2$ in the array $A$ by the binary $4$-tuples $(1000)$, $(0100)$, $(0010)$ and $(0001)$ in the matrix $M$ (of course this map depends on the ordering of $X$ we have chosen, but another choice would lead to a column-permutation-equivalent matrix, hence a permutation-equivalent code). Notice that in matrix $M$, coordinates lying in the same group of the transversal design have been distinguished by dashed vertical lines.

  Matrix $M$ then defines, over any extension $\FF_{2^e}$ of the prime field $\FF_2$, the dual code of the so-called incidence code $\IC_{2^e}(\calC_0)$. For all values of $e$, the incidence codes $\IC_{2^e}(\calC_0)$ have the same generator matrix of $2$-rank $7$, being:
  \[
  G = 
  \left(\begin{array}{c:c:c:c}
          1 0 0 1 & 0 0 0 0 & 0 0 1 1 & 1 0 1 0 \\
          0 1 0 1 & 0 0 0 0 & 0 1 1 0 & 0 0 1 1 \\
          0 0 1 1 & 0 0 0 0 & 0 1 0 1 & 0 1 1 0 \\
          0 0 0 0 & 1 0 0 1 & 0 1 0 1 & 1 1 0 0 \\
          0 0 0 0 & 0 1 0 1 & 0 0 1 1 & 0 1 1 0 \\
          0 0 0 0 & 0 0 1 1 & 0 1 1 0 & 0 1 0 1 \\
          0 0 0 0 & 0 0 0 0 & 1 1 1 1 & 1 1 1 1
  \end{array}\right)\,.
  \]

\end{example}

\subsubsection{A deeper analysis of incidence codes coming from linear MDS codes of dimension $2$} Incidence codes lead to an innumerably large family of PIR protocols --- as many as there exists codes $\calC_0$ --- but most of them are not practical for PIR protocols (essentially because the kernel of the incidence matrix is too small). To simplify their study, one can first remark that intuitively, the more blocks a transversal design, the larger its incidence matrix, and consequently, the lower the dimension of its associated code. But the number of blocks of $\calT_{\calC_0}$ is the cardinality of $\calC_0$. Hence, informally the smaller the code $\calC_0$, the larger $\IC(\calC_0)$.

We recall that a $[n, k, d]$ linear code is said to be maximum distance separable (MDS) if it reaches the Singleton bound $n+1 = k+d$. Besides, the dual code of an MDS code is also MDS, hence its dual distance is $k+1$. In this paragraph we analyse the incidence codes constructed with MDS codes of dimension $2$. Their interest lies in being the smallest codes with dual distance $3$, which is the minimal setting for defining $1$-private PIR protocols.

Generalized Reed-Solomon codes are the best-known family of MDS codes.

\begin{definition}[generalized Reed-Solomon code]
  Let $\ell \ge k \ge 1$. Let also $\mathbf{x} \in \FF_q^\ell$ be a tuple of pairwise distinct so-called \emph{evaluation points}, and $\mathbf{y} \in (\FF_q^\times)^\ell$ be the \emph{column multipliers}. We associate to $\mathbf{x}$ and $\mathbf{y}$ the \emph{generalized Reed-Solomon} (GRS) code:
  \[
  \begin{aligned}
  \mathrm{GRS}_k(\mathbf{x}, \mathbf{y}) \mydef \{ &(y_1f(x_1),\dots,y_\ell f(x_\ell)), \\
  &\quad  f \in \FF_q[X], \deg f < k \}\,.
  \end{aligned}
  \]
\end{definition}

Generalized Reed-Solomon codes $\mathrm{GRS}_k(\mathbf{x},\mathbf{y})$ are linear MDS codes of dimension $k$ over $\FF_q$, and they give usual Reed-Solomon codes when $\mathbf{y} = (1,\dots,1)$. Moreover, GRS codes are essentially the only MDS codes of dimension $2$, as states the following lemma whose proof can be found in the Appendix.

\begin{lemma}
  \label{lem:MDS-dim2}
  All $[\ell, 2, \ell-1]$ MDS codes over $\FF_q$ with $2 \le \ell \le q$ are generalized Reed-Solomon codes.
\end{lemma}

Let us study the consequences of Lemma~\ref{lem:MDS-dim2} in terms of transversal designs. We say a map $\phi: X \to X'$ is an isomorphism between transversal designs $(X, \calB, \calG)$ and $(X', \calB', \calG')$ if it is one-to-one and if it preserves the incidence relations, or in other words, if $\phi$ is invertible on the points, blocks and groups:
\[
\phi(X) = X',\quad \phi(\calB) = \calB',\quad \phi(\calG) = \calG'.
\]

\begin{lemma}
  Let $\calC, \calC'$ be two codes such that $\calC' = \mathbf{y} \ast \calC$ for some $\mathbf{y} \in (\FF_q^{\times})^{\ell}$, where $\ast$ is the coordinate-wise product of $\ell$-tuples. Recall $\calT_{\calC},\, \calT_{\calC'}$ are the transversal designs they respectively define. Then, $\calT_{\calC}$ and $\calT_{\calC'}$ are isomorphic.
\end{lemma}
\begin{proof}
  Write $\calT_{\calC} = (X, \calB, \calG)$ and $\calT_{\calC'} = (X', \calB', \calG')$. From the definition it is clear that $X = X' = \FF_q \times [1, \ell]$ and $\calG = \calG' = \{ \FF_q \times \{i\}, 1 \le i \le \ell\}$. Now consider the blocks sets. We see that $\calB = \{ \{ (c_i, i), 1 \le i \le \ell\}, c \in \calC\}$ and $\calB' = \{ \{ (y_ic_i, i), 1 \le i \le \ell\}, c \in \calC\}$. Let:
  \[
  \begin{array}{rlrl}
    \phi_{\mathbf{y}} : & \FF_q \times [1, \ell] & \to     &\FF_q \times [1, \ell] \\
                     & (x, i)                 & \mapsto & (y_i x, i)
  \end{array}
  \]
  The vector $\mathbf{y}$ is $\ast$-invertible, hence $\phi_{\mathbf{y}}$ is one-to-one on the point set $X$. It remains to notice that $\phi_{\mathbf{y}}$ maps $\calG$ to itself since it only acts on the first coordinate, and that $\phi_{\mathbf{y}}(\calB)$ is exactly $\calB'$ by definition of $\calC$ and $\calC'$.
\end{proof}

\begin{proposition}
  Let $2 \le \ell \le q$ and $\calC_0$ be an $[\ell, 2, \ell - 1]_q$ linear (MDS) code. Let also $\FF_p$ be any finite field. The incidence code $\IC_p(\calC_0)$ is permutation-equivalent to $\IC_p(\mathrm{RS}_2(\mathbf{x}))$, with $\mathbf{x} \in \FF_q^{\ell}$, $x_i \ne x_j$.
\end{proposition}
\begin{proof}  
  Lemma~\ref{lem:MDS-dim2} shows that all $[\ell, 2, \ell-1]_q$ linear codes $\calC_0$ can be written as $\mathbf{y} \ast \mathrm{RS}_2(\mathbf{x})$ for some $\mathbf{x} \in \FF_q^{\ell}$. Moreover, with the previous notation $\phi_{\mathbf{y}}(\calT_{\mathrm{RS}_2(\mathbf{x})}) = \calT_{\mathbf{y} \ast \mathrm{RS}_2(\mathbf{x})}$, so we have $u \in \IC_p(\mathbf{y} \ast \mathrm{RS}_2(\mathbf{x}))$ if and only if $u \in \Code_p(\phi_{\mathbf{y}}(\calT_{\mathrm{RS}_2(\mathbf{x})}))$. Now, let:
  \[
  \begin{array}{rlrl}
    \tilde{\phi}_{\mathbf{y}} : & \FF_p^X           &\to     & \FF_p^X \\
                             & u = (u_x)_{x \in X} &\mapsto & (u_{\phi_{\mathbf{y}}(x)})_{x \in X} 
  \end{array}\,.
  \]
  Clearly $\tilde{\phi}_{\mathbf{y}}(\IC_p(\mathrm{RS}_2(\mathbf{x}))) = \Code_p(\phi_{\mathbf{y}}(\calT_{\mathrm{RS}_2(\mathbf{x})}))$ and $\tilde{\phi}_{\mathbf{y}}$ is a permutation of coordinates. So $\IC_p(\calC_0)$ is permutation-equivalent to $\IC_p(\mathrm{RS}_2(\mathbf{x}))$ which proves the result.
\end{proof}

In our study of incidence codes of $2$-dimensional MDS codes $\calC_0$, the previous proposition allows us to restrict our work on Reed-Solomon codes $\calC_0 = \mathrm{RS}_2(\mathbf{x})$ with $\mathbf{x}$ an $\ell$-tuple on pairwise distinct $\FF_q$-elements.

A first result proves that if $\mathbf{x}$ contains all the elements in $\FF_q$, then $\IC_q(\mathrm{RS}_2(\mathbf{x}))$ is the code which has been previously studied in subsection~\ref{subsec:affine-geometry}. More precisely,
\begin{proposition}
  The following two codes are equal up to permutation:
  \begin{enumerate}
  \item $\calC_1 = \IC_q(\mathrm{RS}_2(\FF_q))$, the incidence code over $\FF_q$ of the full-length Reed-Solomon code of dimension $2$ over $\FF_q$;
  \item $\calC_2$, the code over $\FF_q$ based on the transversal design $\calT_A(2, q)$.
  \end{enumerate}
\end{proposition}

\begin{proof}
  It is sufficient to show that the transversal design defined by $\calC_0 = \mathrm{RS}_2(\FF_q)$ is isomorphic to $\calT_A(2, q)$. Let us enumerate $\FF_q = \{ x_1,\dots,x_q \}$. We recall that $\calT_{\calC_0} = (X, \calB, \calG)$ where:
  \[
  \begin{aligned}
    X &= \FF_q \times [1,q], \\
    \calB &= \{ \{ (c_i, i), i \in [1,q] \}, c \in \calC_0 \}, \\
    \calG &= \{ \{ (\alpha, i), \alpha \in \FF_q \}, i \in [1,q] \}\,,
  \end{aligned}
    \]
  and that $\calT_A(2, q) = (X', \calB', \calG')$ with:
  \[
  \begin{aligned}
    X' &= \FF_q \times \FF_q, \\
    \calB' &= \{ \{ (a x_i + b, x_i), i \in [1,q] \}, (a,b) \in \FF_q^2 \} \\
    \calG' &= \{ \{ (\alpha, x_i), \alpha \in \FF_q \}, i \in [1,q] \}\,.
  \end{aligned}
  \]

  In the light of the above, one defines $\phi : X \to X',\, (\alpha, i) \mapsto (\alpha, x_i)$, which is clearly one-to-one and satisfies $\phi(\calG) = \calG'$. Moreover, a codeword $c \in \calC_0$ is the evaluation of a polynomial of degree $\le 1$ over $\FF_q$. Hence for some $(a,b) \in \FF_q^2$, we have  $c_i = a x_i + b, \forall i$. This proves that $\phi$ extends to a one-to-one map $\calB \to \calB'$, giving the desired isomorphism.
\end{proof}

It remains to study the case of tuples $\mathbf{x}$ of length $\ell < q$. First, one may notice that $\IC_q(\mathrm{RS}_2(\mathbf{x}))$ is a shortening of $\IC_q(\mathrm{RS}_2(\FF_q))$. Indeed, we have the following property:

\begin{lemma}
  Let $\calC_0$ be a linear code of length $\ell$ over $\FF_q$, and $\overline{\calC_0}$ be a puncturing of $\calC_0$ on $s$ positions. Then for all prime powers $q'$, $\IC_{q'}(\overline{\calC_0})$ is a shortening of $\IC_{q'}(\calC_0)$ on the coordinates corresponding to $s$ groups of the transversal design $\calT_{\calC_0}$.
\end{lemma}

\begin{proof}
  Without loss of generality, we can assume that $\calC_0$ is punctured on its $s$ last coordinates in order to give $\overline{\calC_0}$. Let us analyse the link between $\calT_{\overline{\calC_0}} = (\overline{X}, \overline{\calB}, \overline{\calG})$ and $\calT_{\calC_0} = (X, \calB, \calG)$. We have:
  \[
  \begin{array}{rll}
    \overline{X} &= \FF_q \times [1, \ell - s] &\subset X,\\
    \overline{\calG} &= \{ \FF_q \times \{ i \},\, i \in [1, \ell-s] \} &\subset \calG,\\
    \overline{\calB} &= \{ B \cap \overline{X},\, B \in \calB \} &~
  \end{array}
  \]
  Let $\calC = \IC_{q'}(\calC_0)$ and $\overline{\calC} = \IC_{q'}(\overline{\calC_0})$. For clarity, we index words in $\calC$ (resp. $\overline{\calC}$) by $X$ (resp. $\overline{X}$). For $\overline{c} \in \FF_{q'}^{\overline{X}}$, we define $\mathrm{ext}(\overline{c}) \mydef c \in \FF_{q'}^X$, such that $c_{|\overline{X}} = \overline{c}$ and $c_{|X\setminus \overline{X}} = 0$. By definition of code's puncturing/shortening, all we need to prove is:
  \[
  \overline{\calC} = \{ \overline{c} \in \FF_{q'}^{\overline{X}},\, \mathrm{ext}(\overline{c}) \in \calC \}.
  \]
  Remind that $\overline{\calC}$ is defined as the set of $\overline{c} \in \FF_{q'}^{\overline{X}}$ satisfying $\sum_{b \in \overline{B}} \overline{c}_b = 0$ for every $\overline{B} \in \overline{\calB}$. Hence we have:
  \[
  \begin{aligned}
    \overline{c} \in \overline{\calC} &\iff \sum_{b \in \overline{B}} \overline{c}_b = 0, \quad \forall \overline{B} \in \overline{\calB}\\
    &\iff \sum_{b \in B \cap \overline{X}} \overline{c}_b = 0, \quad \forall B \in \calB\\
    &\iff \sum_{b \in B \cap \overline{X}} \mathrm{ext}(\overline{c})_b \\
    &\quad\quad\quad + \, \sum_{b \in B \cap (X \setminus \overline{X})} \mathrm{ext}(\overline{c})_b = 0, \quad \forall B \in \calB\\
    &\iff \sum_{b \in B} \mathrm{ext}(\overline{c})_b = 0, \quad \forall B \in \calB\\
    &\iff \mathrm{ext}(\overline{c}) \in \calC
  \end{aligned}
  \]
  We conclude the proof by pointing out that $X \setminus \overline{X}$ is a union of $s$ distinct groups from $\calG$. 
\end{proof}

Despite this result, incidence codes of Reed-Solomon codes $\mathrm{RS}_2(\mathbf{x})$ remain hard to classify for $|\mathbf{x}| = \ell < q$. Indeed, for a given length $\ell < q$, some $\IC(\mathrm{RS}(\mathbf{x}))$ appear to be non-equivalent. Their dimension can even be different, as shows an exhaustive search on $\IC_{16}(\mathrm{RS}(\mathbf{x}))$ with pairwise distinct $\mathbf{x} \in \FF_q^\ell$, $q=16$ and $\ell = 5$: we observe that $48$ of these codes have dimension $24$ while the $4 320$ others have dimension $22$. Further interesting research would then be to understand the values of $\mathbf{x}$ leading to the largest codes, for a fixed length $|\mathbf{x}| = \ell$.

\subsection{High-rate incidence codes from divisible codes}
\label{subsec:divisible-codes}

In this subsection, we prove that linear codes $\calC_0$ satisfying a \emph{divisibility} condition yield incidence codes whose rate is roughly greater than $1/2$. Let us first define divisible codes.

\begin{definition}[divisibility of a code]
 Let $p \ge 2$. A linear code is $p$-divisible if $p$ divides the Hamming weight of all its codewords.
\end{definition}

A study of the incidence matrix which defines an incidence code leads to the following property.
\begin{lemma}
  \label{lem:equation-MM}
  Let $\calC_0$ be a code of length $\ell$ over a set $S$, and let $\calT$ be the transversal design associated to $\calC_0$. We denote by $M$ the incidence matrix of $\calT$, where rows of $M$ are indexed by codewords from $\calC_0$. Then we have:
  \[
  (M M^T)_{c, c'} = \ell - d(c, c')\quad \forall c,c' \in \calC_0\,,
  \]
  where $d(\cdot,\cdot)$ denotes the Hamming distance.
\end{lemma}

\begin{proof}
  For clarity we adopt the notation $M[c,(\alpha,i)]$ for the entry of $M$ which is indexed by the codeword $c \in \calC_0$ (for the row), and $(\alpha, i) \in S \times [1,\ell]$ (for the column). We also denote by $\mathds{1}_{\mathcal{U}(c, i, \alpha)} \in \{0,1\}$ the boolean value of the property $\mathcal{U}$, that is, $\mathds{1}_{\mathcal{U}(c, i, \alpha)} = 1$ if and only if $\mathcal{U}(c, i, \alpha)$ is satisfied. Now, let $c, c' \in \calC_0$.
  \[
  \begin{aligned}
    (M M^T)_{c, c'} &= \sum_{\alpha \in S,\, i \in [1,\ell]} M[c, (\alpha,i)] M[c', (\alpha,i)] \\
    &= \sum_{\alpha \in S,\, i \in [1,\ell]} \mathds{1}_{c_i = \alpha} \mathds{1}_{c'_i = \alpha} \\
    &= \sum_{i=1}^{\ell} \sum_{\alpha \in S} \mathds{1}_{c_i = c'_i = \alpha} \\
    &= \sum_{i=1}^\ell \mathds{1}_{c_i = c'_i}\\
    &= \ell - d(c,c')\,.
  \end{aligned}
  \]
\end{proof}

Hence, if some prime $p$ divides $\ell$ as well as the weight of all the codewords in $\calC_0$, then the product $MM^T$ vanishes over any extension of $\FF_p$, and $M$ is a parity-check matrix of a code containing its dual. A more general setting is analyzed in the following proposition.

\begin{proposition}\label{prop:p-divisible}
  Let $\calC_0$ be a linear code of length $\ell$ over $S$, $|S| = s$. Let also $\calC = \IC_q(\calC_0)$ with $\mathrm{char}(\FF_q) = p$. Denote the length of $\calC$ by $n = \ell s$. If $\calC_0$ is $p$-divisible, then
  \[
  \calC^\perp \cap \parity \subseteq \calC\,,
  \]
  where $\parity$ denotes the parity-check code of length $n$ over $\FF_q$. In particular, we get $\dim \calC \ge \frac{n - 1}{2}$.
  
  Moreover, if $p \mid \ell$, then $\calC^\perp \subseteq \calC$ and $\dim \calC \ge \frac{n}{2}$.
\end{proposition}

\begin{proof}
  Let $M$ be the incidence matrix of the transversal design $\calT_{\calC_0}$. Also denote by $J$ and $J'$ the all-ones matrices of respective size $|\calC_0| \times n$ and $|\calC_0| \times |\calC_0|$. If we assume that $\calC_0$ is $p$-divisible, then Lemma~\ref{lem:equation-MM} translates into 
  \begin{equation}
    \label{eq:MM}
  M M^T = \ell J' \mod p
  \end{equation}
  while an easy computation shows that
  \[
  M J^T = \ell J'\,.
  \]
  Hence, over $\FF_q$ we obtain
  \begin{equation}
    \label{eq:MMJ}
  M (M - J)^T = 0
  \end{equation}
  which brings us to consider the code $A$ of length $n$ generated over $\FF_q$ by the matrix $M - J$. Equation~\eqref{eq:MMJ} indicates that $A \subseteq \calC$. Let $\parity \mydef \{ c \in \FF_q^n, \sum_i c_i = 0 \}$ be the parity-check code of length $n$ over $\FF_q$. Notice that $c \in \parity \iff cJ^T = 0$ and $uJ = 0 \iff uJ' = 0$. If $p \nmid \ell$, this leads to:
  \[
  \begin{aligned}
    \calC^\perp \cap \parity &= \{ c = uM \in \FF_q^n, cJ^T = 0 \}\\
    &= \{ c = uM \in \FF_q^n, \ell u J' = 0 \} \\
    &= \{ c = uM \in \FF_q^n, u J = 0 \} \\
    &= \{ u(M - J) \in \FF_q^n, u J = 0 \} \subseteq A  \subseteq \calC \,.
  \end{aligned}
  \]
  On the other hand, if $p \mid \ell$, then equation~\eqref{eq:MM} turns into $MM^T = 0$, meaning that $\calC^\perp \subseteq \calC$.

  Finally, the first bound on the dimension comes from
  \[
  \dim \calC \ge \dim (\calC^\perp \cap \parity) \ge \dim \calC^\perp -1 = n - \dim \calC - 1\,,
  \]
  while the second one is straightforward.
\end{proof}

In terms of PIR protocols, previous result translates into the following corollary.

\begin{corollary}
  Let $p$ be a prime, and assume there exists a $p$-divisible linear code of length $\ell_0$ over $\FF_q$. Then, there exists $k \ge (\ell_0 q - 1)/2$ such that we can build a distributed PIR protocol for a $k$-entries database over $\FF_q$, and whose parameters are $\ell(k) = \ell_0$ and $n(k) = \ell_0 q \le 2k+1$.
\end{corollary}

Divisible codes over small fields have been well-studied, and contain for instance the extended Golay codes~\cite[ch.II.6]{MacWilliamsS77}, or the famous MDS codes of dimension $3$ and length $q+2$ over $\FF_q$~\cite[ch.XI.6]{MacWilliamsS77}.

\begin{example}\label{ex:golay}
  The extended binary Golay code is a self-dual $[24, 12, 8]_2$ linear code. It produces a transversal design with $24$ groups, each storing $2$ points. Its associated incidence code $\Code_2(\mathrm{Golay})$ has length $n = 24 \times 2 = 48$ and dimension $\ge 24$, and by computation we can show that this bound is tight.
\end{example}

\begin{remark}\label{rem:binary-golay}
  In our application for PIR protocols, we would like to find divisible codes $\calC_0$ defined over large alphabets (compared to the code length), but these two constraints seem to be inconsistent. For instance, the binary Golay code presented in Example~\ref{ex:golay} leads to a PIR protocol with a too expensive communication cost ($24$ bits of communication for an original file of size... $24$ bits: that is exactly the communication cost of the trivial PIR protocol where the whole database is downloaded). Nevertheless, Example~\ref{ex:golay} represents the worst possible case for our construction, in a sense that the rate of $\IC_2(\mathrm{Golay}_2)$ is exactly $1/2$ (it attains the lower bound), and that each server stores $2$ bits (which is the smallest possible). Codes with better rate and/or with larger server storage capability would then give PIR protocols with relevant communication complexity. For instance, the extended ternary Golay code gives better parameters --- see Example~\ref{ex:golay3}.
\end{remark}

Divisible codes over large fields seems not to have been thoroughly studied (to the best of our knowledge), since coding theorists use to consider codes over small alphabets as more practical. We hope that our construction of PIR protocols based on divisible codes may encourage research in this direction.

\section{PIR protocols with better privacy}
\label{sec:collusion-PIR}

When servers are colluding, the PIR protocol based on a simple transversal design does not ensure a sufficient privacy, because the knowledge of two points on a block gives some information on it. To solve this issue, we propose to use orthogonal arrays with higher strength $t$.

\subsection{Generic construction and analysis}

In the previous section, classical ($t=2$) orthogonal arrays were used to build transversal designs. Considering higher values of $t$, we naturally generalize the latter as follows:

\begin{definition}[$t$-transversal designs]
  Let $\ell \ge t \ge 1$. A $t$-transversal design is a block design $\calD = (X, \calB)$ equipped with a group set $\calG = \{ G_1, \dots,  G_\ell \}$ partitioning $X$ such that:
  \begin{itemize}
  \item $|X| = s \ell$;
  \item any group has size $s$ and any block has size $\ell$;
  \item for any $T \subseteq [1,\ell]$ with $|T| = t$ and for any $(x_1,\hdots,x_t) \in G_{T_1} \times \hdots \times G_{T_t}$, there exist exactly $\lambda$ blocks $B \in \calB$ such that $\{x_1,\hdots,x_t\} \subset B$.
  \end{itemize}
  A $t$-transversal design with parameters $s, \ell, t, \lambda$ is denoted $\tTD_{\lambda}(\ell, s)$, or $\tTD(\ell, s)$ if $\lambda = 1$.
\end{definition}

Given a $t$-transversal design $\mathcal{T}$, we can build a $(t-1)$-private PIR protocol with the exactly the same steps as in section~\ref{sec:no-collusion-PIR}. First, we define the code $\calC = \Code_q(\calT)$ associated to the design according to Definition~\ref{def:code-of-a-design}, and then we follow the algorithm given in Figure~\ref{fig:no-collusion-PIR}. Since a $t$-transversal design is also a $2$-transversal design for $t \ge 2$, the analysis is identical for every PIR feature, except for the security where it remains very similar.

\textbf{Security ($(t-1)$-privacy).} Let $T$ be a collusion of servers of size $|T| \le t-1$. For varying $i \in I$, the distributions $\calQ(i)_{|T}$ are the same because there are exactly $\lambda s^{t -1 -|T|} \ge \lambda \ne 0$ blocks which contain both $i$ and the queries known by the servers in $T$.

To sum up, the following theorem holds:
\begin{theorem}
  \label{thm:t-private-pir}
  Let $D$ be a database with $k$ entries over $\FF_q$, and $\calT = \tTD(\ell, s)$ be a $t$-transversal design, whose incidence matrix has rank $\ell s - k$ over $\FF_q$. Then, there exists an $\ell$-server $(t-1)$-private PIR protocol with:
  \begin{itemize}
  \item only $1$ symbol to read for each server,
  \item $\ell-1$ field operations for the user,
  \item $\ell \log(sq)$ bits of communication,
  \item a (total) storage overhead of $(\ell s - k) \log q$ bits on the servers.
  \end{itemize}
\end{theorem}

\subsection{Instances and results}
\label{subsec:t-private-instances}

\subsubsection{$t$-transversal designs from curves of degree $\le t-1$}

Looking for instances of $t$-transversal designs, it is natural to try to generalise the transversal designs of Construction~\ref{cons:td-affine}. An idea is to turn affine lines into higher degree curves.

\begin{construction}
  \label{cons:degree-t}
Let $X$ be the set of points in the affine plane $\FF_q^2$, and $\calG = \{G_1, \dots, G_q\}$ be a partition of $X$ in $q$ parallel lines. W.l.o.g. we choose the following partition: $G_i = \{ (y, \alpha_i), y \in \FF_q\}$ for each $\alpha_i \in \FF_q$. Blocks are now defined as the sets of the form
\[
B_F = \{ (F(x), x), x \in \FF_q \},\,\text{ where } F \in \FF_q[x],\,\deg F \le t-1.
\]
\end{construction}

\begin{lemma}
  The design $(X, \calB, \calG)$ given in Construction~\ref{cons:degree-t} forms a $t$-transversal design $\tTD_1(q, q)$.
\end{lemma}

\begin{proof}
  The group set indeed partitions $X$ into $q$ groups, each of size $q$. It remains to check the incidence property. Let $\{G_{T_1},\dots,G_{T_t}\}$ be a set of $t$ distinct groups, and let $((y_{T_1}, x_{T_1}), \dots, (y_{T_t}, x_{T_t})) \in G_{T_1} \times \dots \times G_{T_t}$. From Lagrange interpolation theorem, we know there exists a unique polynomial $F \in \FF_q[X]$ of degree $\le t-1$ such that:
\[
F(x_{T_j}) = y_{T_j}\,\quad \forall 1 \le j \le t\,.
\]
Said differently, there is a unique block which contains the $t$ points $\{ (y_{T_j}, x_{T_j}) \}_{1 \le j \le t}$.
\end{proof}

We do not yet analyse the rank properties of these designs, since Construction~\ref{cons:degree-t} corresponds to a particular case of the generic construction given below.

\subsubsection{$t$-transversal designs from orthogonal arrays of strength $t$}

In this paragraph we give a generic construction of $t$-transversal designs, which is a simple generalisation of the way we build transversal designs with orthogonal arrays (Subsection~\ref{subsec:oa-to-td}).

\begin{construction}
  \label{cons:ttd-from-toa}
  Let $A$ be an orthogonal array $\OA_{\lambda}(t, \ell, s)$ on a symbol set $S$. Recall that the array $A$ is composed of rows $a_i = (a_{i,j})_{1 \le j \le \ell}$ for $1 \le i \le \lambda s^t$. We define the following design:
  \begin{itemize}
  \item its point set is $X = S \times [1,\ell]$;
  \item its group set is $\calG = \{ S \times \{i\}, 1 \le i \le \ell \}$;
  \item its blocks are $B_i = \{ (a_{i,j}, j), 1 \le j \le \ell \}$ for all $a_i \in \mathrm{Rows}(A)$.
  \end{itemize}
\end{construction}

\begin{proposition}
  \label{prop:t-td}
  If $A$ is an $\OA_{\lambda}(t, \ell, s)$, then the design defined with $A$ by Construction~\ref{cons:ttd-from-toa} is a $\tTD_{\lambda}(\ell, s)$.
\end{proposition}
\begin{proof}
 It is clear that $\calG$ is a partition of $X$ and that blocks and groups have the claimed size. Now focus on the incidence property. Let $T \subset [1,\ell]$ with $|T| = t$, and let $(x_1,\hdots,x_t) \in G_{T_1} \times \hdots \times G_{T_t}$. We need to prove that there are exactly $\lambda$ blocks $B \in \calB$ such that $\{x_1,\hdots,x_t\} \subset B$.

Consider the map from blocks in $\calB$ to rows of $A$ given by:
\[
\begin{array}{rclc}
  \psi : &\calB                                      &\to     & \mathrm{Rows}(A)\\
         & B_i = \{ (a_{i,j}, j) , 1 \le j \le \ell \} &\mapsto & (a_{i,1},\hdots,a_{i,\ell})
\end{array}
\]
Since we assumed that orthogonal arrays have no repeated row, the map $\psi$ is one-to-one. Denote by $x' = (x'_1,,\hdots,x'_t) \in S^t$ the vector formed by the first coordinates of $(x_1,,\hdots,x_t) \in X^t$. From the definition of an orthogonal array of strength $t$ and index $\lambda$, we know that $x'$ appears exactly $\lambda$ times in the submatrix of $A$ defined by the columns indexed by $T$. Hence this defines $\lambda$ preimages in $\calB$, which proves the result.
\end{proof}

\begin{remark}
  As we noticed before, Construction~\ref{cons:degree-t} is a particular case of Construction~\ref{cons:ttd-from-toa}. Indeed, a block $B_F = \{(F(x), x), x \in \FF_q \}$, with $\deg F \le t-1$ is in one-to-one correspondence with a codeword $c_F$ of a Reed-Solomon code of dimension $t$.
\end{remark}

\begin{corollary}
  \label{coro:dual-distance-t-private}
  Let $\calC_0$ be a code of length $\ell$ and dual distance $t+2 \le \ell$ over a set $S$ of size $s$. Then, $\IC_q(\calC_0)$ defines a $t$-private PIR protocol.
\end{corollary}
\begin{proof}
  Let $A$ be the orthogonal array defined by $\calC_0$. We know that $A$ has strength $t+1$ (see \emph{e.g.}~\cite{MacWilliamsS77}), hence from Proposition~\ref{prop:t-td}, the associated transversal design is a $\genTD{(t+1)}(\ell, s)$. Theorem~\ref{thm:t-private-pir} then ensures that the PIR protocol induced by this transversal design is $t$-private.
\end{proof}

As in Section~\ref{sec:explicit-constructions}, if the code $\calC_0$ is divisible, then we can give a lower bound on the rate of its incidence code. We provide two examples in finite (and small) length.

\begin{example}
\label{ex:golay3}
A first example would be to consider extended Golay codes. Indeed, they are known to be divisible by their characteristic~\cite[ch.II.6]{MacWilliamsS77}, they have large dual distance, and Proposition~\ref{prop:p-divisible} then ensures their incidence codes have non-trivial rate. In Remark~\ref{rem:binary-golay}, we noticed that the binary Golay code does not lead to a practical PIR protocol due to a large communication complexity. Thus, let us instead consider the $[12, 6, 6]_3$ extended ternary Golay code, that we denote $\mathrm{Golay_3}$. It is self-dual, hence $d^\perp(\mathrm{Golay_3}) = 6$. Then, $\calC = \IC_{3^e}(\mathrm{Golay}_3)$, $e \ge 1$, has length $36$ and Proposition~\ref{prop:p-divisible} shows that $\dim \calC \ge 18$ (the bound can be proved to be tight by computation). Hence, the associated PIR protocol works on a raw file of $18$ $\FF_{3^e}$-symbols encoded into $36$, uses $12$ servers (each storing $3$ $\FF_{3^e}$-symbols) and resists any collusion of one third (\emph{i.e.} $4$) of them.

\end{example}

\begin{example}
  A second example arises from the exceptional $[q+2, 3, q]_q$ MDS codes in characteristic $2$~\cite[ch.XI.6]{MacWilliamsS77}. For instance, for $q = 4$, we obtain a $2$-private PIR protocol with $6$ servers, each storing $4$ symbols of $\FF_{2^e}$ for some $e \ge 1$. Once again, the dimension of the incidence code attains the lower bound, here $k =12$.
\end{example}

\begin{example}
  Examples of incidence codes which do not attain the lower bound of Proposition~\ref{prop:p-divisible} come from binary Reed-Muller codes of order $1$, denoted $\mathrm{RM}_2(m, 1)$. These codes are $2$-divisible since they are known to be equivalent to extended Hamming codes. They also have length $n = 2^m$ and dual distance $d^\perp = n/2$.

  For instance, $\mathrm{RM}_2(3, 1)$ provides an incidence code of dimension $k = 11 > 8$, that is, a $2$-private $8$-server PIR protocol on a database with $11$ $\FF_{2^e}$-symbols, where each server stores $2$ symbols. For $m=4$, $\mathrm{RM}_2(4, 1)$ gives a $6$-private $16$-server PIR protocol on a database with $20$ $\FF_{2^e}$-symbols, each server storing $2$ symbols. We conjecture that $\mathrm{IC}_2(\mathrm{RM}_2(m, 1))$ leads to a $(2^{m-1}-2)$-private $2^m$-server PIR protocol on a database with $2^m + m$ symbols, each server storing $2$ symbols.
\end{example}

As pointed out in Subsection~\ref{subsec:oa-to-td}, high-rate incidence codes $\calC = \IC(\calC_0)$ have the best chance to occur when the dimension of $\calC_0$ is small, since the cardinality of $\calC_0$ is the number of rows in a (non-full-rank) parity-check matrix which defines $\calC$. Besides, in order to define $t$-private PIR protocols, we need an orthogonal array of strength $t+2$, \emph{i.e.} a code $\calC_0$ with dual distance $t+2$. Conciliating both constraints, we are tempted to pick MDS codes of dimension $t+1$.

A well-known family of MDS codes is the family of Reed-Solomon codes. For $\calC_0 = \mathrm{RS}_{t+1}(\FF_q)$ and varying values of $q$ and $t$, we were able to compute the rate of $\IC(\calC_0)$, and these codes lead to $t$-private PIR protocols with communication complexity approximately $\sqrt{n}$, where $n$ is the length of the encoded database. These rates are presented in Figure~\ref{fig:rates-t-private} and as expected, the rate of our families of incidence codes decreases with $t$, the privacy parameter. Figure~\ref{fig:rates-t-private} also shows that Reed-Solomon-based instances cannot expect to reach at the same time constant information rate and resistance to a constant fraction of colluding servers.

\begin{figure}[h!]
  \centering
  \includegraphics{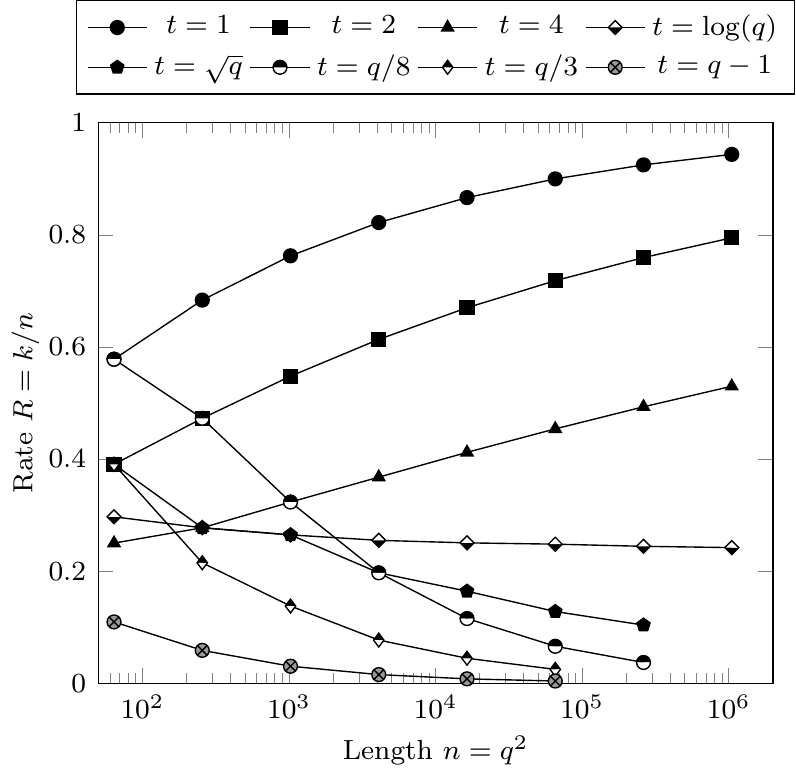}
  %  \begin{tikzpicture}[scale=1]
  %   \begin{axis}[
  %     xmode=log, xmin=50, xmax=2000000, ymin=0, ymax=1,
  %     cycle list name=mark list*,
  %     legend style={at={(0.5,1.05)},anchor=south,legend columns=4,font=\footnotesize},
  %     yticklabel style = {font=\footnotesize},
  %     xticklabel style = {font=\footnotesize},
  %     ylabel style = {at={(0.11,0.5)},anchor=south,font=\footnotesize},
  %     xlabel style = {at={(0.5,0.02)},anchor=north,font=\footnotesize},
  %     xlabel={Length $n = q^2$}, ylabel={Rate $R = k/n$}
  %     ]
  %     \addplot table[x index=2, y index=4] {../data/t-private-t1.dat};
  %     \addlegendentry{$t=1$}

  %     \addplot table[x index=2, y index=4] {../data/t-private-t2.dat};
  %     \addlegendentry{$t=2$}

  %     \addplot table[x index=2, y index=4] {../data/t-private-t4.dat};
  %     \addlegendentry{$t=4$}

  %     \addplot table[x index=2, y index=4] {../data/t-private-tlogq.dat};
  %     \addlegendentry{$t=\log(q)$}

  %     \addplot table[x index=2, y index=4] {../data/t-private-tsqrtq.dat};
  %     \addlegendentry{$t=\sqrt{q}$}

  %     \addplot table[x index=2, y index=4] {../data/t-private-tq8.dat};
  %     \addlegendentry{$t=q/8$}

  %     \addplot table[x index=2, y index=4] {../data/t-private-tq3.dat};
  %     \addlegendentry{$t=q/3$}

  %     \addplot table[x index=2, y index=4] {../data/t-private-tqminus1.dat};
  %     \addlegendentry{$t=q-1$}

  %   \end{axis}
  % \end{tikzpicture}
  \caption{\label{fig:rates-t-private}Rate of incidence codes of $\calC_0$ that are used for building $t$-private PIR protocols. Codes $\calC_0$ are full-length Reed-Solomon codes of dimension $t+1$ (dual distance $t+2$) over $\FF_q$. Associated PIR protocols then need $q$ servers, each storing $q$ symbols.} 
\end{figure}

\section{Comparison with other works}
\label{sec:comparison}

Our construction fits into the model of distributed (or coded) PIR protocols, which is currently instantiated in a few schemes, notably the construction of Augot~\emph{et al.}~\cite{AugotLS14} and all the works involving the use of PIR codes initiated by Fazeli~\emph{et al}~\cite{FazeliVY15}. We recall that we aimed at building PIR protocols with very low burden for the servers, in terms of storage \emph{and} computation. While PIR codes are a very efficient way to reduce the storage overhead, they do not cut down the computation complexity of the original replication-based PIR protocol  used for the emulation.

Hence, for the sake of consistency, we will only compare the parameters of our PIR schemes with those of the multiplicity code construction presented in~\cite{AugotLS14}.

\textbf{Sketch of the construction~\cite{AugotLS14}.} Multiplicity codes $\calC$ have the property that a codeword $c \in \calC$ can be seen as the vector of evaluations of a multivariate polynomial $f_c \in \FF_q[X_1,\dots,X_m]$ and its derivatives over the space $\FF_q^m$, where $\FF_q$ denotes the finite field with $q$ elements. Every affine line of the space $\FF_q^m$ then induces linear relations between $f_c$ and its derivatives, which translates into low-weight parity-check equations for the codewords. This allows to define a local decoder for $\calC$: when trying to retrieve a symbol $D_i$ indexed by $i \in \FF_q^m$, one can pick random affine lines going through $i$ and recover $D_i$ by computing short linear combinations of the symbols associated to the evaluations of $f_c$ and their derivatives along these lines. We refer to~\cite{KoppartySY14} for more details on these codes.

Augot \emph{et al.}~\cite{AugotLS14} realized that partitioning $\FF_q^m$ into $q$ parallel hyperplanes gives rise to storage improvements. By splitting the encoded database according to these hyperplanes and giving one part to each of the $q$ servers, they obtained a huge cut on both the total storage and the number of servers, while keeping an acceptable communication complexity. Their construction requires a minor modification of the LDC-based PIR protocol of Figure~\ref{fig:LDC-to-encoded-PIR}; indeed, in the query generation process, the only server which holds the desired symbol must receive a random query. Nevertheless, the PIR scheme they built was at that time the only one to let the servers store less than twice the size of the database. Moreover, the precomputation of the encoding of the database ensures an optimal computational complexity for the servers. As noticed previously, we emphasize the significance of this feature when the database is very frequently queried.

\textbf{Parameters of the distributed PIR protocol~\cite{AugotLS14} based on multiplicity codes.} This PIR scheme depends on four main parameters: the field size $q$, the dimension $m$ of the underlying affine space, the multiplicity order $s$ and the maximal degree $d$ of evaluated polynomials. For an error-free and collusion-free setting, $d = s(q-1) -1$ is the optimal choice. Let $\sigma = \binom{m+s-1}{m}$. The associated PIR protocol uses $q$ servers to store an original database containing $\binom{d+m}{m}$ $\FF_q$-symbols, but encoded into codewords of length $q^m$, where each symbol has size $\sigma \log q$ bits. Hence the redundancy (in bits) of the scheme is:
\[
\rho = \left(\sigma q^m - \binom{s(q-1) + m - 1}{m} \right) \log q\,.
\]
Concerning the communication complexity, let us only focus on the download cost (which is often the bottleneck in practice). The multiplicity code local decoding algorithm needs to query symbols of $\sigma$ distinct lines of the space. Hence, each server must answer $\sigma$ symbols of size $\sigma \log q$ bits. Thus it leads to a download communication complexity of 
\[
\gamma =  \sigma^2 q \log q \,\,\text{ bits}.
\]

\textbf{Note about the comparison strategy.} We consider a database $D$ of size $100$ MB. Since the protocols may be initially constructed for databases of smaller size $k$, we split $D$ into $k$ chunks of size $|D|/k$. Hence, when running the PIR protocol, the user is allowed to retrieve a whole chunk, and the chunk size will be precised in our tables. For instance, in the first row of Table~\ref{tab:comp-64}, one shall understand that the user is able to retrieve $31.1$kB of the database privately, with $1.99$MB of communication, while each server only produces $1$ operation over the chunks (of size $31.1$kB) it holds. 

\begin{table*}[t!]
  \[
  \begin{array}{ccccc}
    \text{Instance}               & \text{\shortstack{download\\communication}} & \text{\shortstack{complexity\\(\#op./server)}} & \text{storage overhead} & \text{chunk size} \\
    \hline
    \calT_A(m=2, q=64)            & 1.99\,\mathrm{MB} &  1 & 22.7\,\mathrm{MB} & 31.1\,\mathrm{kB} \\
    \calT_A(m=3, q=64)            &   56\,\mathrm{kB} &  1 &  126\,\mathrm{MB} &   882\,\mathrm{B} \\
    \mathrm{Mult}(q=64, m=3, s=6) & 2.32\,\mathrm{MB} & 56 & 64.8\,\mathrm{MB} &   16\,\mathrm{B}  \\
    \mathrm{Mult}(q=64, m=4, s=2) &   15\,\mathrm{kB} &  5 &  694\,\mathrm{MB} &   13\,\mathrm{B}  \\
    \hline
  \end{array}
  \]
  \caption{\label{tab:comp-64}Comparison of some distributed PIR protocols with $64$ servers on a $100$MB initial database. Parameters of multiplicity codes have been chosen in order to obtain simultaneously low communication complexity and storage overhead.}
\end{table*}

\begin{table*}[t!]
  \[
  \begin{array}{ccccc}
    \text{Instance}               & \text{\shortstack{download\\communication}} & \text{\shortstack{complexity\\(\#op./server)}} & \text{storage overhead} & \text{chunk size} \\
    \hline
    \calT_A(m=2, q=8)            & 22.7\,\mathrm{MB}  & 1 & 76\,\mathrm{MB}  & 2.83\,\mathrm{MB}   \\
    \calT_A(m=3, q=8)            & 6.03\,\mathrm{MB} & 1  & 281\,\mathrm{MB}  & 754\,\mathrm{kB} \\
    \mathrm{Mult}(q=8, m=4, s=2) & 8.8\,\mathrm{MB} & 5 & 797\,\mathrm{MB} & 117\,\mathrm{kB} \\
    \mathrm{Mult}(q=8, m=6, s=3) & 2.86\,\mathrm{MB} & 28  & 3.24\,\mathrm{GB}  & 1.2\,\mathrm{kB} \\
    \hline
  \end{array}
  \]
  \caption{\label{tab:comp-8}Comparison of some distributed PIR protocols with $8$ servers on a $100$MB initial database. We notice that for $q=8$ servers, there only exist a few non-Reed-Muller ($s \ge 2$) multiplicity codes whose associated PIR protocols have communication complexity strictly less than the size of the original database.}
\end{table*}

Tables~\ref{tab:comp-64} and~\ref{tab:comp-8} first reveal that our PIR schemes are more storage efficient than the PIR schemes relying on multiplicity codes. Moreover, our constructions provide a better \emph{communication rate} (defined as the ratio between communication cost and chunk size), though the multiplicity code PIR protocols allow to retrieve smaller chunks (hence is more flexible). 

\begin{remark}
  Recent constructions of PIR protocols (for instance results of Sun and Jafar such that~\cite{SunJ17}) lead to better parameters in terms of communication complexity. However, we once more emphasize that we aimed at minimizing the computation carried out by the servers, which is a feature that is mostly not considered in those works.
\end{remark}

\section{Conclusion}

In this work, we have presented a generic construction of codes yielding distributed PIR protocols with optimal server computational complexity, in a sense that each server only has to read one symbol of the part of the database it stores. Our construction makes use of transversal designs, whose incidence properties ensure a natural distribution of the coded database on the servers, as well as the privacy of the queries. Our PIR protocols also feature efficient reconstructing steps since the user has to compute a simple linear combination of the symbols it receives. Finally, they require low storage for the servers and acceptable communication complexity.

We instantiated our construction with classical transversal designs coming from affine and projective geometries, and with transversal designs emerging from orthogonal arrays of strength $2$. The last construction that we call \emph{incidence code} can even be generalized, since stronger orthogonal arrays lead to PIR protocols with a better resilience to collusions.

The generality of our construction allows the user to choose appropriate settings according to the context (low storage capability, few colluding servers, \emph{etc.}). It also raises the question of finding transversal designs with the most practical PIR parameters for a given context. Indeed, while affine and projective geometries give excellent PIR parameters for the servers (low computation, low storage), there seems to remain room for improving the communication complexity and the number of needed servers.

\bibliographystyle{plain}
\bibliography{pir_final_version}{}

\appendix

\subsection{Hamada's formula}
\label{app:hamada}

Hamada~\cite{Hamada68} gives a generic formula to compute the $p$-rank of a projective geometry design $\PG_t(m,q)$, for $q = p^e$:
\[
\small
\begin{aligned}
&\rank_p(\PG_t(m, q)) \\
&\quad = \sum_{(s_0,\hdots,\,s_e) \in S} \prod_{j=0}^{e-1} \sum_{i=0}^{L(s_{j+1},s_j)} (-1)^i\textstyle \binom{m+1}{i} \textstyle\binom{m+ s_{j+1}p - s_j - ip}{m}
\end{aligned}
\]
 where $S \subset \ZZ^{e+1}$ contains elements $(s_0,\hdots,\,s_e)$ such that:
\[
\left\{
\begin{array}{l}
  s_0 = s_e\\
  t+1 \le s_j \le m+1\\
  0 \le s_{j+1} p - s_j \le (m+1)(p-1)\,,
\end{array}\right.
\]
and $L(s_{j+1},s_j) = \lfloor \frac{s_{j+1} p - s_j}{p} \rfloor$ .

The $p$-rank of the associated affine geometry design $\AG_t(m,q)$ can be derived from the projective one by:
\[
\begin{aligned}
  &\rank_p(\AG_t(m, q)) \\
  &\quad = \rank_p(\PG_t(m,q)) - \rank_p(\PG_t(m-1,q))\,.
  \end{aligned}
\]

 Despite its heavy expression, Hamada's formula can be simplified by picking very specific values of $m$, $p$ or $e$. For instance we have:
\[
\begin{aligned}
  m = 2  &: \quad \forall p, e, \rank_p\, \AG_1(2, p^e) = \textstyle\binom{p+1}{2}^e\,,\\
  e = 1  &: \quad \forall p, m,  \rank_p\, \AG_1(m, p) = p^m - \textstyle\binom{m+p-2}{m}\,.
\end{aligned}
\]
For $(m,e)  = (3,2)$, we get
\[
\forall p, \quad \rank_p\, \AG_1(3, p^2) = \left(p^3 - \textstyle\binom{p+1}{3}\right)^2 + 2 \textstyle\binom{p}{2}\binom{p+1}{3}\,,
\]
this equality being found by interpolation, since $\rank_p(\AG_1(m, p^e))$ is a polynomial of degree at most $m e$ in $p$.

\subsection{Proof of Lemma~\ref{lem:MDS-dim2}}

Let us recall the result we want to state.

\begin{lemma*}
  All $[\ell, 2, \ell-1]$ MDS codes over $\FF_q$ with $2 \le \ell \le q$ are generalized Reed-Solomon codes.
\end{lemma*}

\begin{proof}
  First we know that GRS codes are MDS.

  Let $\calC$ be an $[\ell, 2, \ell-1]_q$ code with $2 \le \ell \le q$. Since $\calC$ is MDS, it has dual distance $d^\perp = 3$, and we claim there exists a codeword $c \in \calC$ with Hamming weight $\ell$. Indeed, let $G = (P_1, \dots, P_\ell)$ be a generator matrix of $\calC$, where $P_i \in \FF_q^2$ is written in column. Notice that each point $P_i$ is non-zero (otherwise $d^\perp = 1$) and $0, P_i, P_j$ are not on the same line for $i \ne j$ (otherwise $d^\perp = 2$). Moreover codewords in $\calC$ are simply evaluations of bilinear maps $\mu: \FF_q^2 \to \FF_q$ over $(P_1,\dots,P_\ell)$:
\[
\calC = \{ (\mu(P_1), \dots, \mu(P_\ell)), \mu \in \mathcal{L}(\FF_q^2, \FF_q) \}\,,
\]
and the $P_i$'s are not all on the same line (otherwise, $\dim \calC \le 1$). 

Since $\ell \le q$, there exists $Q = (Q_0, Q_1) \in \FF_q^2 \setminus \{0\}$ such that $Q$ does not lie in the (vector) line defined by any of the $P_i$'s. Let now $\mu_Q(X,Y) = Q_1X - Q_0Y$: it is a non-zero bilinear form which must vanish on a line of $\FF_q^2$, and since $\mu_Q(Q) = 0$, it vanishes on the one spanned by $Q$. To sum up, for every $i \in [1, \ell]$, we have $\mu_Q(P_i) \ne 0$. Hence, $c = (\mu_Q(P_1),\dots,\mu_Q(P_\ell))$ belongs to $\calC$ and has Hamming weight $\ell$.
  
  Let now $u \in \calC$ such that $\{c, u\}$ spans $\calC$. We denote by $c \ast u$ the coordinate-wise product $(c_1 u_1,\dots,c_{\ell}u_{\ell})$ and by $\mathbf{1}$ the all-one vector of length $\ell$. Then $c = \mathbf{1} \ast c$ and  $u = c \ast (c^{-1} \ast u)$, where $c^{-1}$ is the coordinate-wise inverse of $c$ through $\ast$. Hence, the code $\calC$ can be written $c \ast \calC'$ where $\calC'$ has $G' = \binom{\mathbf{1}}{c^{-1} \ast u}$ as generator matrix. It means that $\calC$ is the GRS code with evaluation points $\mathbf{x} = c^{-1} \ast u$, multipliers $\mathbf{y} = c$ and dimension $2$.
\end{proof}

\section*{Acknowledgments}

% This work is partially funded by  French ANR-15-CE39-0013-01 \enquote{Manta}.
%The author is very grateful to Daniel Augot and Françoise Levy-dit-Vehel for their very helpful comments on the paper.
The author would like to thank Françoise Levy-dit-Vehel and Daniel Augot for their valuable comments, and more specifically the first collaborator for her helpful guidance all along the writing of the paper.

% \begin{IEEEbiographynophoto}{Julien Lavauzelle}
% received his engineer's degree at the École Nationale Supérieure des Techniques Avancées (ENSTA, Palaiseau, France) in 2014, where he specialized in information systems and their security. In 2015, he received a master's degree in theoretical computer science at the École Normale Supérieure, Cachan. Since 2015, he is PhD student in the Crypto team hosted by the laboratory of computer science of École Polytechnique and by Inria (Palaiseau, France), under the supervision of Françoise Levy-dit-Vehel and Daniel Augot. His main interests concern the use of codes with locality in cryptographic applications, such that private information retrieval protocols and proofs of retrievability.
% \end{IEEEbiographynophoto}

\end{document}